\documentclass[conference]{IEEEtran}
\IEEEoverridecommandlockouts
\usepackage{cite}
\usepackage{amsmath,amssymb,amsfonts}

\usepackage{graphicx}
\usepackage{textcomp}
\usepackage{xcolor}
\def\BibTeX{{\rm B\kern-.05em{\sc i\kern-.025em b}\kern-.08em
		T\kern-.1667em\lower.7ex\hbox{E}\kern-.125emX}}

\usepackage{multirow}
\usepackage{amsfonts}
\usepackage{amsthm}
\usepackage{amssymb}
\usepackage{mathtools}

\usepackage[subrefformat=parens,labelformat=parens]{subfig}
\usepackage[export]{adjustbox}
\usepackage[linesnumbered,ruled]{algorithm2e}
\SetAlCapSty{}

\usepackage{algpseudocode}

\usepackage{cleveref}
\usepackage[numbers]{natbib}

\usepackage{mathptmx} 

\usepackage{etoolbox}

\algdef{SxnE}[VARIABLES]{Variables}{EndVariables}
{\textbf{Variables:}}
\algdef{SxnE}[CONSTANTS]{Constants}{EndConstants}
{\textbf{Constants:}}
\algdef{SxnE}[RULES]{Rules}{EndRules}
{\textbf{Rules:}}
\algdef{SxnE}[PREDICATES]{Predicates}{EndPredicates}
{\textbf{Predicates:}}
\usepackage{comment}
\usepackage{soul}
\usepackage{color}

\newcommand\hlbreakable[1]{\textcolor{red}{#1}}

\crefname{lemma}{Lemma}{Lemmas}
\crefname{theorem}{Theorem}{Theorems}
\crefname{definition}{Definition}{Definitions}

\newtheorem{theorem}{Theorem}

\newtheorem{lemma}{Lemma}
\newtheorem{definition}{Definition}

\newtheorem*{question}{Question}
\SetAlFnt{\small\sffamily}

\newcommand{\red}[1]{\textcolor{black}{#1}}

\begin{document}

\title{A Game-Theoretic Approach to Self-Stabilization with Selfish Agents}

\author{%
	{Amir Reza Ramtin{\small$^{1}$}, Don Towsley{\small $^{2}$}}%
	\vspace{1.6mm}\\
	\fontsize{10}{10}\selectfont\itshape
	College of Information and Computer Sciences, University of Massachusetts Amherst\\
	Amherst, USA\\
	\fontsize{9}{9}\selectfont\ttfamily\upshape
	%
	$^{1}$\,aramtin@cs.umass.edu\\
	$^{2}$\,towsley@cs.umass.edu%
}

\maketitle

\begin{abstract}
		Self-stabilization is an excellent approach for adding fault tolerance to a distributed intelligent system. However, two properties of self-stabilization theory, convergence and closure, may not be satisfied if agents are selfish. 
		To guarantee convergence, we formulate the problem as a stochastic Bayesian game and introduce probabilistic self-stabilization to adjust the probabilities of rules with behavior strategies. This satisfies agents' self-interests such that no agent deviates the rules. To guarantee closure in the presence of selfish agents, we propose fault-containment as a method to constrain legitimate configurations of the self-stabilizing system to be Nash equilibria. We also assume selfish agents as capable of performing unauthorized actions at any time, which threatens both properties, and present a stepwise solution to handle it. As a case study, we consider the problem of distributed clustering and propose five self-stabilizing algorithms for forming clusters. Simulation results show that our algorithms react correctly to rule deviations and outperform comparable schemes in terms of fairness and stabilization time.
\end{abstract}


\maketitle


\section{Introduction}
\label{section:introduction}
A distributed intelligent system (DIS) is a 
	network of intelligent agents that interact and exchange data to solve complex problems. 

Being distributed, a DIS is subject to transient faults (e.g., unpredictable changes in the behaviors/roles of agents) due to temporary hardware, software, and communication failures, which can result in the failure of the whole system. Therefore, there is a strong need for fault-tolerant approaches to detect and tolerate transient faults.

Much of the existing literature on handling faults in DISs is based on traditional fault-tolerant schemes for distributed systems \cite{arfat_survey_2016}. However, because each DIS has its own set of specifications and characteristics, no scheme is suitable for all situations.
We focus on a specific
  setup where agents are spatially distributed and connected through a wireless ad-hoc network (e.g., distributed robotics, sensor networks, and Internet of Things)
 and show that self-stabilization \cite{dijkstra1974self} is a promising approach for fault tolerance in the collective behavior \cite{rossi2018review} of agents.

\red{Self-stabilization characterizes the ability of a distributed algorithm to converge in a finite time to a configuration from which it behaves correctly, regardless of the arbitrary initial configuration of the system. Under such an algorithm, the system} recovers from a transient fault regardless of its nature. Consequently, a self-stabilizing DIS is fault tolerant regardless of the dynamics of the environment. Moreover, a distributed self-stabilizing algorithm can be used to initialize a distributed system so that it eventually ends in a legitimate configuration, regardless of its initial configuration. In relation to a DIS, we assume a configuration is legitimate if it allows agents to achieve a common goal.

A DIS can consist of either non-cooperative agents that act selfishly, each agent maximizing its own gain from the interaction, or cooperative agents working to achieve a common goal \cite{grubshtein2012partial}, or both. In the first scenario, selfishness can prevent the system from correctly self-stabilizing. Thus the rules for self-stabilization need to account for selfishness. In a self-stabilizing system, agents are expected to comply with the algorithm and never intentionally perform unwanted actions. However, such behavior may not occur when agents act selfishly.
Refusal to execute a self-stabilizing rule and the intermittent execution of a false action are two adverse consequences of selfishness. \red{We can classify these,} respectively, as permanent and intermittent faults. However, typically, a self-stabilizing algorithm only guarantees recovery from a transient fault (i.e. a fault that occurs once and then disappears).

The first paper that addresses the issue of selfishness in self-stabilization is \cite{dasgupta2006selfish}. The authors assume selfish agents sometimes cooperate while maintaining limited forms of self-interest. In the context of self-stabilization, this behavior shows up as agents preferring different subsets of legitimate configurations. 
The main problem with their proposed method is that 
the legitimate configuration and the resulting Nash equilibrium \cite{nash1950equilibrium} must be unique, which is very restrictive and difficult to achieve.
Similarly, in \cite{gouda_nash_2009, ramtin_self-stabilizing_2014}, the authors 
proposed approaches that discourages selfish nodes
from perturbing legitimate configurations. However, none of these approaches prevent them from deviating from the algorithm during convergence.

In this paper, our goal is to answer the following question.

\begin{question}
	How do we design a self-stabilizing algorithm for a distributed system given that agents may deviate from the algorithm during or after convergence because of their private goals?
\end{question}

We categorize the types of deviations that selfish agents may employ in a self-stabilizing system, and provide solutions for each category that ensure closure and convergence even with selfish agents.
Briefly, closure requires the system to stay within a set of legitimate configurations
and convergence requires the system to reach a legitimate configuration. 
To address the problem of competition during convergence, we borrow ideas related to probabilistic self-stabilization to model 
agent self-interest by using probabilistic rules. Subsequently, we use a stochastic Bayesian game-theoretic modeling approach to extract these self-interests. 
We illustrate our framework through a case study, the design of a self-stabilizing clustering algorithm that constructs a Maximal Independent Set (MIS) in a network of selfish agents. We prove that our clustering solutions exhibit two desirable properties,  closure and convergence, in spite of the presence of selfish agents.
Similar game-theoretic approaches to clustering \cite{kassan_game_2018} do not provide self-stabilization. The experimental setup of our work includes analyzing the performance of the proposed algorithms through simulation. In summary, the
main contributions of this work include:
\begin{itemize}
	\item A stochastic Bayesian game model for self-stabilization.
	\item A 
	framework for the design of non-cooperative 
	self-stabilizing algorithms.
	\item Self-stabilizing clustering algorithms for systems consisting of selfish agents.
\end{itemize}

The rest of the paper is organized as follows: We introduce the concepts and discuss solution methods in \Cref{section:approach}. In \Cref{section:casestudy}, the proposed algorithms for clustering are presented, and proofs are given to establish their correctness. In \Cref{section:analysisComplexity}, we analyze the complexities of the proposed algorithms. \Cref{section:results} deals with the numerical evaluation of the algorithms, and comparisons are made to contrast their performance against the prior art. The paper ends with conclusions.

\section{Related Work}
\label{section:related}

In this section, we briefly overview the related work to our approach and case study. Regarding selfishness in distributed computing, we only consider some of the most related work.

\subsection{Game-Theoretic Self-Stabilizing Approaches}

Connections between self-stabilization and game theory have been previously recognized. In \cite{jaggard2014self}, the authors related self-stabilization to uncoupled dynamics, a procedure used in game theory to reach a Nash equilibrium in situations where players do not know each other's payoff functions. In turn, \cite{apt2018self} showed how Dijkstra's solution to self-stabilization can be naturally formulated using standard concepts of strategic games, notably the concept of an improvement path. They also showed how one can reason about them in game-theoretic terms. In \cite{yen2016designing}, \red{the authors transformed the game-theoretic models of a problem into self-stabilizing algorithms that achieve an intended system goal through private goals of the processes.}

In \cite{dasgupta2006selfish}, the authors associated a cost function with each node 
and presented a simple self-stabilizing algorithm that constructs a spanning tree 
corresponding to a Nash equilibrium of an underlying strategic game. This method makes the legitimate configuration unique, which is very restrictive. 
In \cite{gouda_nash_2009}, the authors equated a legitimate configuration with an outcome of a game to determine whether or not an agent can benefit from a perturbation. They showed that no agent has an incentive to perturb a legitimate configuration if all legitimate configurations are Nash equilibria.
A similar case-specific solution for selfish perturbations has been discussed in \cite{ramtin_self-stabilizing_2014} where authors handled the self-stabilizing problem of virtual backbone construction in selfish wireless ad-hoc networks. These approaches discourage selfish nodes from perturbing legitimate configurations; however, they do not prevent them from deviating from the algorithm during convergence. 

\subsection{Selfishness in Distributed Computing}


\cite{abraham2013distributed} analyzes leader-election protocols that give rises to an equilibrium (cf. Section 2.1) in different settings. The goal of our case study and that of \cite{abraham2013distributed} are similar although the approaches and applications are different. In our approach, since we target self-stabilization in addition to equilibria in legitimate configurations, we discuss a stochastic game model to tolerate deviations during convergence \red{in any distributed setting.} 

\cite{afek2014distributed} leverages mechanism design in game theory to design principles that penalize agents that deviate from the rules and protocols. The main \red{difference} between our approach and \cite{afek2014distributed} \red{is that, in our approach,}
we tolerate deviations but still guarantee that the system eventually stabilizes.
Furthermore, we model the system as a stochastic game, which eliminates the need for 
wake-up timings \red{required in \cite{afek2014distributed}}. 

\cite{collet2018equilibria} discusses a dynamic Bayesian game model for distributed randomized algorithms where agents are selfish, which shares some similarities with our proposed stochastic Bayesian game modeling approach in terms of how to manage incomplete information and actions; however, it does not address the characteristics of distributed self-stabilization, such as the existence of different local states. 

\subsection{Self-Stabilizing Algorithms for Maximal Independent Set}

The first self-stabilizing algorithms introduced to build an MIS \red{work with a central scheduler} \cite{shukla1995observations,hedetniemi2003self}. In \cite{lin2003efficient}, the authors proposed a fault-containing algorithm that is an improvement on \cite{hedetniemi2003self}.

The early distributed self-stabilizing algorithms for building an MIS 
have time complexity $O(n^2)$ \cite{ikeda2002space, goddard2003self,shi2004anonymous}. \cite{turau_linear_2007} proposed the first self-stabilizing algorithm that has a linear time complexity \red{and works under} an unfair distributed scheduler. In \cite{ramtin_self-stabilizing_2014-1}, the authors proposed a distributed self-stabilizing algorithm that is fault-containing with contamination distance of one. \cite{arapoglu2019energy} introduced a similar algorithm to \cite{turau_linear_2007} that takes at most $\max(3n-6,2n-1)$ moves to stabilize. In \cite{yen2015selfish}, the authors first proposed a mechanism design for the MIS problem and then turned the design into a self-stabilizing algorithm that works slightly better than the previous works. \cite{arapoglu2019asynchronous} proposed a distributed self-stabilizing algorithm that uses two-hop information to build an MIS with unfair distributed scheduler, which stabilizes after $n-1$ moves and is significantly better than the others in terms of move complexity.

\section{Problem Definition and System Model}
\label{section:problem}

\noindent Self-stabilization is an important concept for distributed computing and communication networks. It describes a system's ability to recover automatically from a transient fault, which is crucial for a DIS as its computational power lies in the fact that it is distributed and agents communicate
with their neighbors in the underlying network to fulfill a given task.
An unexpected change in the system, such as loss of communication or failure due to dynamics in the environment, can be modeled as a transient fault. As defined below, we consider self-stabilization of the system as a whole, where no single agent alone can achieve the design objective of the DIS in which it resides \cite{funk2002self}. 

\begin{definition}\label{d1}{\em
		Self-Stabilization \cite{dolev2000self}.}
	Let $\mathbb{P}$ be a system property that identifies the correct execution of a distributed system. A system $Q$ is
	self-stabilizing if the following two conditions are satisfied. a) Starting from an arbitrary state, $Q$ reaches a state where $\mathbb{P}$ holds (convergence), after execution of a finite number of actions. b) Once $\mathbb{P}$ is established for $Q$, then $\mathbb{P}$ remains valid (closure), after any subsequent execution of actions.
\end{definition}

A DIS is specified by 
an undirected connected graph $G=(V,E)$, where $V$ is a set of agents, and $E$ is a set of edges. Let $N(v)$, $v \in V$ denote the set of neighbors of $v$ in $G$, i.e. $N(v)=\{u:(v,u)\in E\}$.

An agent is modeled by a finite set of primary variables, $Vars$.
Each variable $X \in Vars$ has a finite domain $\mathbb{V}_{X}$.
The state of an agent, denoted by $\gamma$, is the realization (values) of its primary variables, i.e., $\gamma \in \prod_{X \in Vars}\mathbb{V}_{X}$. Let $\Gamma$ denote the set of all possible states.
Let $v.\gamma \in \Gamma$ be the state of agent $v\in V$. Every agent can access the states of its neighbors using an asynchronous distributed message passing algorithm.
The \textit{k-local state} of $v$ corresponds to the states of all agents \red{in its $k$-hop neighborhood, including itself. We will focus mostly on $k=1$ and $k=2$.}
A \textit{configuration (global state)} of the system is a vector consisting of the states of each agent (i.e., $c = \prod_{v\in V}v.\gamma$). Let $C$ denote the set of all possible configurations. Configurations satisfying a system property $\mathbb{P}$ (a predicate on the system) are said to be legitimate. Let $L \subseteq C$ denote the set of legitimate configurations.

A distributed self-stabilizing algorithm consists of a finite set of rules per agent referred to as a local program. A rule has the form 
\begin{center}
	$<guard> \longrightarrow <action>$
\end{center}
Here guard is a Boolean predicate on the \red{$k$-local state} of the agent. A rule is said to be enabled if its guard evaluates to true. In this case, the corresponding action may be executed. 
An agent is said to be enabled in configuration $c$ if at least one of its rules is enabled in $c$ \cite{jubran_recurrence_2015,devismes_self-stabilizing_2019}. \red{A \textit{scheduler} (sometimes called a daemon) is a virtual entity (adversary) that models the asynchronism of the system by selecting, at each step, a subset of enabled agents such that only the selected ones are allowed to execute one of their rules. Schedulers are classified based on characteristics such as fairness and distribution. A central scheduler selects exactly one enabled agent at each step. In contrast, a synchronous scheduler chooses at every step every enabled agents. A scheduler is fair if any enabled agent is guaranteed to be selected after a finite number of rounds. A distributed randomized scheduler, which is equivalent to a scheduler under Gouda's strong fairness \cite{devismes2008weak} assumption, is a fair scheduler that in every step selects a non-empty subset of enabled agents with a uniform probability. An unfair distributed scheduler, which is the most challenging one, selects an arbitrary subset of enabled agents at each step. It is only obliged to select a non-empty set if the set of enabled agents is not empty.}
A computation is a transition $c \rightarrow c'$ such that $c,c'\in C$,
which is obtained by the fact that in configuration $c$, a non-empty subset of enabled agents is chosen by a distributed scheduler to atomically execute actions, which takes the system to $c'$.


\Cref{fig:opEx}(a) shows an example of a self-stabilizing system trajectory (i.e., a trajectory in configuration space that represents a sequence of configurations starting from a given initial configuration)  during convergence.

\begin{figure}[!h]
	\centering
	\begin{tabular}{cc}
		\includegraphics[width=0.20\textwidth]{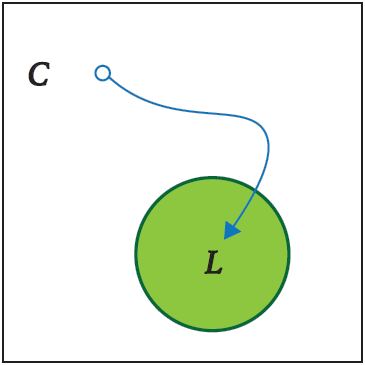}&
		\includegraphics[width=0.20\textwidth]{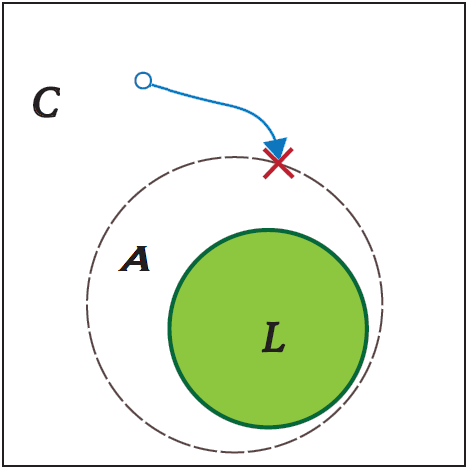}\\
		(a) & (b) \\ 
		\includegraphics[width=0.20\textwidth]{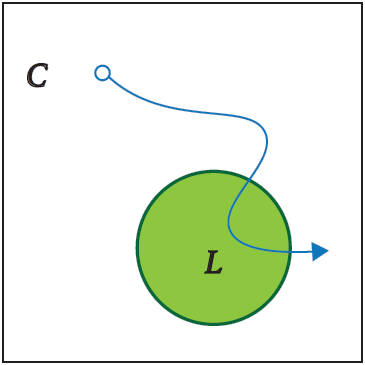} &
		\includegraphics[width=0.20\textwidth]{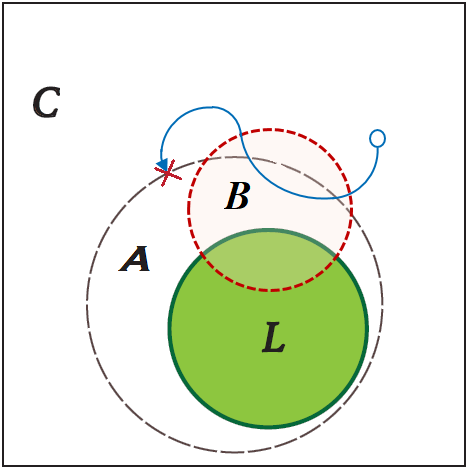} \\
		(c) & (d)
	\end{tabular}
	\caption{An informal presentation of a self-stabilizing system trajectory (note that continuous blue paths (trajectories), each of which represents a sequence of configurations, are in fact discrete) where $C$ is the set of all possible configurations; $L \subset C$ depicts the set of legitimate configurations ($c \in L$ iff $c$ has no enabled rule); $A$ is the set of configurations where any enabled rule is violated due to selfishness, i.e., there is no enabled rule that selfish agents prefer to execute (note that set $L$, any of member of which has no enabled rule, is a subset of $A$); $B$ is the set of configurations where selfish agents execute unauthorized actions; and agents (a) do not deviate, or exhibit selfishness by b) violations, (c) perturbations, or (d) deflections.
	}
	\label{fig:opEx}
\end{figure}

\begin{definition}\label{d7}{\em Gain function.}
	Associated with agent $v \in V$, is a gain function $g_v:C\rightarrow \mathbb{R}^+$ defined on the set of all system configurations. 
\end{definition}


We assume the gain function is such that for each agent $v$ and any illegitimate configuration $c$, there exists a legitimate configuration $c'$ such that the gain of $c'$ is larger than or equal to that of $c$, i.e., $g_v(c')\ge g_v(c)$.

A necessary condition to meet the two criteria of closure and convergence is rule fulfillment. In a self-stabilizing DIS, agents must have the abilities to cooperatively conform to the algorithm \cite{funk2002self} to achieve a common goal corresponding to a configuration that satisfies a desirable system property. This assumption, however, is not always satisfied when agents whose individual goals do not conform with self-stabilizing rules act selfishly.
Even if agents prefer legitimate configurations to illegitimate ones, they may differ as to which one they prefer. This is very similar to the Battle of Sexes game \cite{luce_games_1989} in the case of two agents.

In what follows, we give an example of individual goals and selfish behavior. Consider a sensor network where all sensors have a common goal: they want to form a clustering overlay on the network topology. But in addition to the common goal, each sensor has as its individual goal to minimize its energy consumption while wanting to achieve the common goal. Suppose sensors employ a self-stabilizing algorithm for the clustering problem. Under the assumption that there is neither selfish nor malicious behavior, the self-stabilizing algorithm will achieve a cluster overlay.
Now, assume sensors are aware that cluster-heads consume more energy than regular sensors and thus prefer not to be cluster-heads. As a consequence, they may avoid executing rules that might result in their becoming cluster-heads. This can result in traditional self-stabilization failing to fulfill convergence and closure criteria. 



Note that a self-stabilizing algorithm is an independent process running on an agent; its primary/secondary variables are usually read-only and not writable by other processes running on the agent; and after the system converges to a legitimate configuration, no action of the self-stabilizing algorithm will be executed unless a fault occurs. The agent may have other variables and other executing processes; 
however, these other processes are not allowed to change the variables in $Vars$,
i.e., variables in $Vars$ should be protected from being modified by any action other than the actions of the self-stabilizing rules. Recall that a selfish agent may desire to change those variables. One way of doing so is for the selfish agent to restart (i.e., reinitialize) the process executing the self-stabilizing algorithm. Similarly, the agent can violate the rules by suspending that process.
In the remainder of this section, we discuss three adverse effects of selfishness on self-stabilization, i.e., three types of deviations from cooperative behavior.

\paragraph{Violation}

In a self-stabilizing system, a selfish agent may deliberately avoid executing enabled rules during convergence. \red{We call such rules that agents may prefer to not execute \textit{violation-prone} rules.}
The concept of violation, where selfish agents may not execute enabled rules, is equivalent to the assumption that some rules may never be evaluated, contradicting
the convergence criteria. The fact that selfish agents have individual goals that interfere with the common goal can cause the system to fail to achieve the common goal.
\Cref{fig:opEx}(b) illustrates a scenario where agents may violate rules as the system converges, and finally, the system halts in an illegitimate configuration where all enabled rules are ones that agents violate. Note that in the figure, $A$ is the set of configurations in which either no rules are enabled or, if any, are violated. 

\paragraph{Perturbation}

Here, we introduce another possible deviation, called perturbation, which corresponds to a change in the values of primary variables of an agent in a legitimate configuration that causes an illegitimate one.
A selfish agent may perturb a legitimate configuration (e.g. by restarting itself) in order to increase its gain by forcing the system to converge to a different legitimate configuration. Such 
a selfish action
produces a perturbation.
\Cref{fig:opEx}(c) shows a scenario where a self-stabilizing system, starting from an illegitimate configuration,  reaches a legitimate one, and then is disrupted by a perturbation that returns it to an illegitimate configuration.

\paragraph{Deflection}

The most crucial type of deviation is a deflection, which means regardless of the self-stabilizing rules, a selfish agent is able to a) intentionally execute or b) purposefully ignore any action that assigns new values to its primary variables. Noting that executing any action other than the actions of enabled rules is unauthorized in the context of self-stabilization, a deflection occurs when an agent executes an unauthorized action or does not execute a required action. 
\Cref{fig:opEx}(d) illustrates a system trajectory which is deflected due to unauthorized actions and violations committed during convergence. Eventually, it ends in an illegitimate configuration where there is no longer any unauthorized action or enabled rule other than the rules that agents are violating (dashed).

\section{Methodological Approach}
\label{section:approach}

We begin this section with a description of an approach for tolerating violations. We then present a game-theoretic model for self-stabilizing systems. We leverage this model in our proposed approach to adapt to selfish behaviors during convergence. Next, we explain a method that prevents perturbations in a self-stabilizing system. Finally, we introduce a 3-step solution for designing self-stabilizing algorithms in the face of selfishness.

\subsection{Violation Tolerance during Convergence}


The crucial aspect of modeling a self-stabilizing DIS in the presence of violations is that rule fulfillment is not guaranteed, i.e., each agent has a choice about whether or not to execute an enabled action. Even if agents have incentives to fulfill the rules, their incentives may not be equal due to their individual goals.

We can model incentives using mathematical models of conflict and cooperation between intelligent rational decision-makers, i.e., selfish agents. Every incentive to execute an action will be a best response from an agent's perspective given the incentives of others. Because rule fulfillment is a strategic choice, we define an incentive as a probability of executing an action. 

As a self-stabilizing DIS, 
 a rule given a configuration is equivalent to an action in a multi-stage game, and a probability distribution over the rules is likewise equivalent to a behavior strategy in that game \red{as we will discuss in the next section.} We concentrate on the concept of probabilistic self-stabilization \cite{herman1990probabilistic} to leverage probabilistic rules. \red{
 	 Under a randomized scheduler, a probabilistic self-stabilizing system eventually reaches a legitimate configuration with probability one.
}



In each configuration, every agent assigns a probability to its enabled \red{\textit{violation-prone}} rule equal to the behavior strategy of such an action that determines whether or not the agent executes that enabled rule.
By definition, a behavior strategy assigns a probability distribution over the set of possible actions. Since a behavior strategy implies agent self-interest, we expect this randomization of rules to comply with the likelihood that agents do not violate rules;
in other words, the system will be violation tolerant.

Here, we define the concept of Nash equilibrium in self-stabilizing systems, which is necessary to introduce a condition that must be met so that a violation tolerant system guarantees convergence in the face of violations.

\begin{definition}\label{df:nashs}{\em Nash equilibrium in self-stabilization.}
	Let $A^c_{v}$ denotes the set of available actions to agent $v\in V$ in configuration $c\in C$. Then, $c$ is a Nash equilibrium if for every agent $v\in V$, there is no action $a \in A_v^c$ such that, if executed unilaterally can lead to a sequence of configurations $(c,\dots,c')$ where the utility of $v$ in $c'$ is more than that in $c$. 
\end{definition}

Note that in a self-stabilizing system, even if an agent fails to achieve higher payoff in the next configuration by unilaterally executing an action, there may still be a sequence of configurations that ultimately benefits the agent. Therefore, we define a Nash equilibrium in a self-stabilizing system with respect to a sequence of configurations.

\begin{theorem}\label{t6}
	Assume that a self-stabilizing algorithm for a DIS $Q$ exists under the assumption that no agent violates the rules. Then, under the different assumption that agents can violate the rules, so long as there is an illegitimate configuration that is a Nash equilibrium in $Q$, the algorithm cannot guarantee the convergence property.
\end{theorem}

\begin{proof} Suppose that the agents can violate the rules and that there is an illegitimate configuration $c$ that is a Nash equilibrium. Let $c$ be the initial configuration. Due to Nash equilibrium, no enabled agent has an incentive to change its state in $c$, so the convergence property will not be satisfied, i.e., starting from some initial configurations, $Q$ may never reach a legitimate one, which contradicts the fact that $Q$ is self-stabilizing.
\end{proof}

Consider a dynamic game where players are the agents of a self-stabilizing DIS, a payoff \red{(reward)}, i.e., the payout a player receives from the outcome of the game, is the difference between the gains of two consecutive configurations for an agent, and an action is a choice between violating or not violating an enabled rule. Then, \Cref{th:viloation} provides the conditions under which a self-stabilizing system is guaranteed to tolerate violations.

\begin{theorem}\label{t7}
	\label{th:viloation}
	If all Nash equilibria are legitimate configurations and no agent can selfishly change the value of its primary variables, then convergence is guaranteed in the face of violations if the enabled rules of each agent are executed with probabilities corresponding to the agent's behavior strategy assigned to the actions of those rules.
\end{theorem}

\begin{proof}
	\red{As we will discuss in the next section,} the behavior strategies of an equilibrium represent the probabilities that agents rationally execute enabled rules. This implies that agents are behaving rationally in decision making. Because a probabilistic self-stabilizing system eventually converges to a legitimate configuration, the convergence criteria holds.\end{proof}

Next, we present a game theoretic model for self-stabilizing systems. This allows us to obtain behavior strategies associated with the probabilistic rules of an algorithm that tolerates selfish behaviors.

\subsection{Game-Theoretic Model of Self-stabilizing Systems}
\label{subsec:gamemodel}

Game theory is a natural framework for modeling distributed self-stabilizing systems where players (agents), which are not aware of the entire
network topology, alternate their actions in each configuration based on local knowledge and the dynamic game among agents is repeated many times until the system terminates in a legitimate configuration.

We assume rational and intelligent agents that are intent on maximizing their individual utilities. These agents are players of a game of a self-stabilizing DIS. Such a game consists of a sequence of stage games that are played in different rounds where a round represents the minimum unit of time during which the scheduler selects a subset of agents to simultaneously execute actions, and then the next configuration is determined according to these taken actions and the current configuration.

Noting that a distributed scheduler is defined by a probability distribution that depends on the latest actions and the current configuration, a game that models a distributed self-stabilizing system is a stochastic game \cite{shapley_stochastic_1953}, i.e., a dynamic game with probabilistic transitions. Moreover, the particular view of a game, where an agent that has incomplete information chooses an action, and then the payoff of each agent is determined by the payoffs of its neighbors is very similar to games on networks \cite{galeotti_network_2010}. Note that a 
network 
game (a game where players are connected via a network structure) is a Bayesian game \cite{harsanyi1967games, Jackson_2008}. Therefore, we can conclude that any game in a distributed self-stabilizing system is a stochastic Bayesian game \cite{albrecht_belief_2016} (i.e., a multi-stage Bayesian game with probabilistic transitions).




The underlying assumption is that each agent knows its state, the states of all agents that it can communicate with, and the induced subgraph formed from itself and those agents. This private information of an agent corresponds to the notation of type in Bayesian games. Subsequently, beliefs of an agent describe the uncertainty of that agent of the types of the other agents. Note that in a self-stabilizing system, agents need to generate beliefs about the states and localities of agents about whom they have incomplete information (the network beyond their neighborhoods) to model the game despite the fact that they only care about the actions of their neighbors.

We assume that agents observe their true $k$-local state in every stage of the game. Without this assumption, each agent would have to make decision under uncertainty as a result of partially observable $k$-local states (e.g., due to limited and noisy connectivity) and update its belief in the $k$-local state in addition to that of agents beyond its neighborhood \cite{wray_distributed_nodate}. In other words, the game model would be a partially observable stochastic game which is intractable. Another consideration is that in a particular case where the distributed scheduler is synchronous (i.e., all agents are selected in every round), given the current $k$-local state and the actions taken, the next $k$-local state is realized with certainty and thus the dynamic model of the game is not stochastic.


We assume rounds (time periods) $t=0,1,\dots$, where within each round $t$, every agent updates its variables based on the observable information in that round \cite{jaggard_dynamics_2017}. We refer to the base game that occurs during round $t$ as \textit{stage game} $t$. Let $V$ be the finite set of agents indexed $1,\dots,n$.
Assume an agent $v\in V$ that can communicate with all other agents within $k\geqslant 1$ hops. We denote the set of $v$ and its one to $k$-hop neighbors by $\mathcal{L}_v(k)$.
In particular, $\mathcal{L}_v(k)$ specifies the set of agents that $v$ can read their variables. Moreover, $v$ knows the induced subgraph of $\mathcal{L}_v(k)$. We also denote the set of $k$-hop neighbors of $v$ by $\mathcal{L}_v^*(k)\subset \mathcal{L}_v(k)$. In fact, $\mathcal{L}_v^*(k)$ specifies the agents that $v$ has incomplete information about, i.e, the agents that $v$ can read their variables but it does not know their one-hop neighbors. Note that $\mathcal{L}_{v}(k)=\cup_{i=1}^k\mathcal{L}_{v}^*(i) \cup \{v\}$. Figure \ref{fig:graph} shows an example of $\mathcal{L}_{v}(k)$ and $\mathcal{L}_{v}^*(k)$ when $k=2$.
In the following, unless otherwise stated, neighbors will refer to one-hop neighbors.

\begin{figure}
	\centering
	\includegraphics[width=0.35\linewidth]{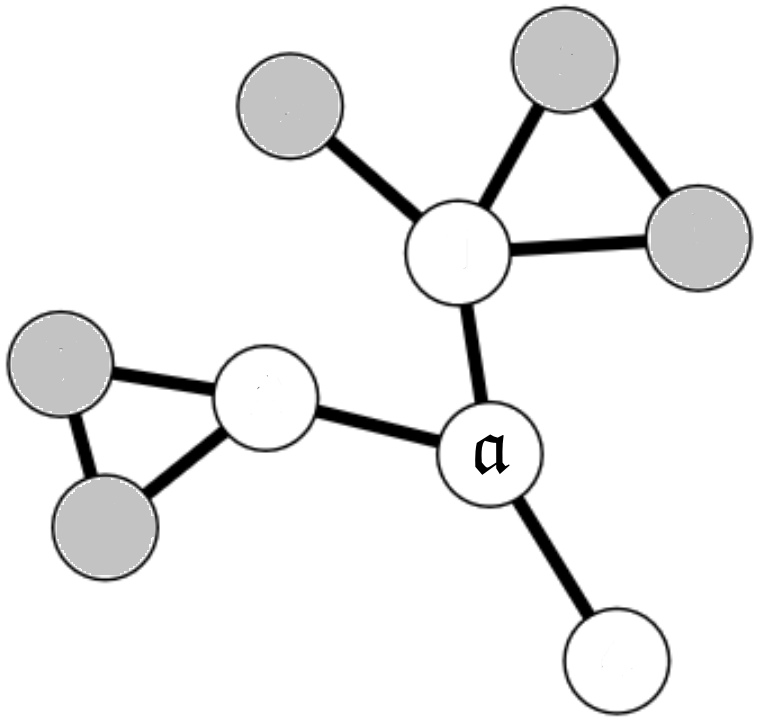}
	\captionof{figure}{A sample subgraph which shows the two-hop neighborhood of agent $\mathfrak{a}$, i.e., $\mathcal{L}_\mathfrak{a}$ given $k=2$. Here, the gray nodes represent the set of agents about whom $\mathfrak{a}$ has incomplete information, i.e., $\mathcal{L}_\mathfrak{a}^*$.}
	\label{fig:graph}
\end{figure}

We assume that the duration of a stage game is long enough for each agent to communicate with all of its neighboring agents within distance $k$ by message passing. We define a $k$-local state of agent $v$, denoted by $\lambda_v(k)$, as the set of primary variables of all agents in $\mathcal{L}_v(k)$. 
Then, an observation of agent $v$, denoted by $o_v(k)$, is a temporary record of variables of every agent in $\mathcal{L}_v(k)$ that $v$ collects via messages in each round. If communications are reliable, $v$ can extract $\lambda_v(k)$ from $o_v(k)$ accurately; otherwise, it needs to estimate $\lambda_v(k)$ from $o_v(k)$. A new stage game starts when all agents update their $k$-local state.
We then characterize the next behavior of each agent through a stochastic Bayesian game framework. 

In this paper, as stated before, we assume reliable communications.
We also assume the gain and available actions of each agent depend only on its own state and that of its one-hop neighbors. 
Finally, we assume that at the end of each round, agents inform each other about the actions taken. Note that agents cannot reliably determine actions by observing local situations alone because not all actions are actually executed due to distributed (asynchronous) scheduling.

Suppose that agent $\mathfrak{a}$ models the game and that $\mathfrak{a}$ communicates with any agent within distance $k$. Our proposed stochastic Bayesian game 
that accounts for both distributed scheduler and incomplete information is $\left< \mathcal{L}_{\mathfrak{a}}(k),\Lambda_{\mathfrak{a}}(k),\Theta,A,g,\mathcal{P},\mu,\delta\right>$, where
\begin{itemize}
	\item $\mathcal{L}_{\mathfrak{a}}(k)=\{w\in V\mid dist(\mathfrak{a},w)\le k\}$ and $\mathcal{L}_{\mathfrak{a}}^*(k)=\{w\in \mathcal{L}_{\mathfrak{a}}(k)\mid dist(\mathfrak{a},w)=k\}$,
	\item $\Lambda_{\mathfrak{a}}(k)$ denotes the set of $k$-local states of $\mathfrak{a}$, 
	\item $\Theta$ is a set of types for every agent, 
	where $\theta_u \in \Theta$ denotes the type of $u\in \mathcal{L}_{\mathfrak{a}}^*(k)$ from point of view of $\mathfrak{a}$, and $\theta=\left<\theta_1,\dots,\theta_{|\mathcal{L}_{\mathfrak{a}}^*(k)|}\right>$ \red{(we do not consider the type of $u \in \mathcal{L}_{\mathfrak{a}}(k)\setminus \mathcal{L}_{\mathfrak{a}}^*(k)$ because under our base assumptions, it is known to $\mathfrak{a}$)},
	\item $A$ is a finite set of actions, where $A_v{(\lambda_{\mathfrak{a}}(k),\theta)}\in A$ is the set of available actions  for agent $v \in \mathcal{L}_{\mathfrak{a}}(k)$ given $k$-local state $\lambda_{\mathfrak{a}}(k)$ and joint type $\theta$,
	\item $g_v: (\Lambda_\mathfrak{a}(k),\Theta^{|\mathcal{L}_{\mathfrak{a}}^*(k)|}) \rightarrow \mathbb{R}$ is the gain function of agent $v \in \mathcal{L}_{\mathfrak{a}}(k)$, 
	\item $\mathcal{P}$ is a set of conditional transition probabilities between $k$-local states, where $\mathcal{P}^(\lambda'_{\mathfrak{a}}(k)\mid \lambda_{\mathfrak{a}}(k), \vec{a})$ denotes the probability of transition from $\lambda(k)$ to $\lambda'_{\mathfrak{a}}(k)$ after taking joint action $\vec{a}$ such that $\lambda_{\mathfrak{a}}(k), \lambda'_{\mathfrak{a}}(k)\in \Lambda_{\mathfrak{a}}(k)$, and $\vec{a}\in \prod_{v\in \mathcal{L}_{\mathfrak{a}}(k)}A$,
	\item $\mu$ is a set of beliefs in types, where $\mu_{u}\left(\theta_{u} \mid h_{\mathfrak{a}}(k,t)\right)$ is the probability that in round $t$, the type of agent $u\in\mathcal{L}_{\mathfrak{a}}^*(k)$ is $\theta_{u}$ given the history of actions and $k$-local states at round $t$ from the point of view of $\mathfrak{a}$, i.e., $h_{\mathfrak{a}}(t)$,
		\item $\delta$ is a discount factor, where $0<\delta\le1$.
\end{itemize}


We assume that from the point of view of each agent, e.g., $v$, agents that are $k$ hops apart maintain private information.
The reason for this is that these agents act based on the states of their neighbors but $v$ does not have access to them. A type $\theta_u \in \Theta$ represents necessary information about the state of sub-graph branched from agent $u \in \mathcal{L}_{v}^*(k)$. The set of types available to an agent depends to the algorithm and specifies its behavior. Note that any agent knows the type of its $k'$-distance neighbors where, $k'<k$, because we have assumed that every agent can read the state of each agent within distance $k$.


\red{To clarify the role of types with an example, we consider a self-stabilizing algorithm for the coloring problem. The algorithm in the case that a node has the same color as a nearby node randomly chooses a color from the set of all possible colors, excluding the color of the node's neighbors but including the color of its own.} Let the colors be numbered from $1$ to $d_{\mathrm{max}}+1$, where $d_{\mathrm{max}}$ is the maximum degree of the graph. Assume each node only knows its one-hop neighbors ($k=1$) and it also has the ability to violate the algorithm rules. Moreover, assume each node gains a profit equal to the reverse of the number of its color if the node has no neighbor of its own color; otherwise, it gains zero. Now, consider the neighborhood of node $v$ as depicted on Figure \ref{fig:graph2} and let $d_{\mathrm{max}}=4$. Then, $v$ can either violate or execute the rule that randomly chooses a color from $\left\{1,2,4,5\right\}$. The decision of $v$ depends on its neighbors, in particular $w$ and $u$. For example, if $w$ has only one neighbor with color of $1$, $v$ may prefer to violate because it knows that $w$ cannot choose $1$ as its color and that there is a chance that the color of $u$ does not change but $w$ sets its color to 5, so $v$ will have to choose a color from $\{1,2,4\}$. Therefore, we conclude that the type of a node in this example is defined as a set of colors specifying its neighbors' colors, which is known only to itself. Hence, nodes have to make beliefs about the types of other nodes.
\begin{figure}[h]
	\centering
	\includegraphics[width=0.35\linewidth]{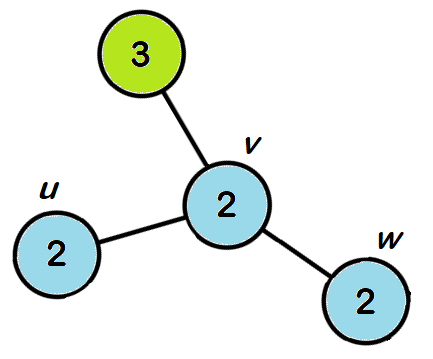}
	\captionof{figure}{\red{An example of a graph coloring problem.}}
	\label{fig:graph2}
\end{figure}

A belief $\mu_{u}\left(\theta_{u} \mid h^{t}\right)$ is the probability that agent $u$'s type is $\theta_{u}$ given history of the game at the beginning of each stage game $t$.
We define history of a game from the point of view of $\mathfrak{a}$ as $h_{\mathfrak{a}}(t)=\left(\lambda_{\mathfrak{a}}^0(k), \vec a^0, \ldots, \lambda_{\mathfrak{a}}^{t-1}(k),\vec a^{t-1}, \lambda_{\mathfrak{a}}^{t}(k)\right)$, where $\lambda_{\mathfrak{a}}^t(k)$ and $\vec a^t$ are respectively the $k$-local state and joint action at round $t$ \red{(superscript $t$ stands for the round number). From now on, for the sake of simplicity, we drop $\mathfrak{a}$ and $k$ from notations $h_{\mathfrak{a}}(k,t)$, $\lambda_{\mathfrak{a}}(k)$, and $\Lambda_{\mathfrak{a}}(k)$ as they can be inferred from notation $\mathcal{L}_{\mathfrak{a}}(k)$, and we follow notation $h^t$ for the history.} We assume that at the end of each stage game, every agent monitors its $k$-local state transitions to
infer the actions taken by other agents
and use this information to update its belief about $\theta$ (types of agents) regarding that stage game. The belief system defines updating the beliefs using Bayes' rule at the end of each stage game.
The posterior belief at the end of stage game $t$ is
\begin{equation}
\label{eq:belief1}
\mu_{u}(\theta_{u}^t \mid a_u, h^t)=\frac{\mu_{u}(\theta_{u}^t \mid h^t) \Pr(a_{u} \mid h^t, \theta_{u}^t)}{\sum_{\tilde{\theta}_{u}^t\in \Theta} \mu_{u}\left(\tilde{\theta}_{u}^t \mid h^{t}\right) \Pr\left(a_{u} \mid h^{t}, \tilde{\theta}_{u}^t\right)},
\end{equation}
where $\mu_u\left(\tilde{\theta}_{u} \mid h^t\right)>0$,
and
$\Pr\left(a_u^t \mid h^t, \theta_{u}\right)>0$ is the probability that action $a_{u}$ is executed at stage game $t$.

Next, we need to update the beliefs regarding the next stage. \red{We do so because the state of agents is not stabilized in the early rounds.} Note that the next condition must be met so that we can solve the game.
\begin{equation}
\label{eq:stabtprim}
\exists t'\ge 0 \forall t\ge t': \Pr(\theta_{u}^{t+1} \mid \theta_u^t, h^{t+1})=\begin{cases}
1, &\theta_{u}^{t+1}=\theta_{u}^{t},\\
0, &\text{otherwise.}
\end{cases}
\end{equation}
Of course, a self-stabilizing system, which stabilizes after a finite number of rounds, guarantees this condition.

The beliefs regarding the next stage are computed as follows:
\begin{equation}
\label{eq:belief2}
\mu_{u}(\theta_u^{t+1} \mid h^{t+1})=\sum_{\theta_u^t\in\Theta}\Pr(\theta_{u}^{t+1} \mid \theta_u^t, h^{t+1})\mu_{u}(\theta_{u}^t \mid h^{t+1}), 
\end{equation}
where $\mu_{u}(\theta_{u}^t \mid h^{t+1})$ is the belief regarding stage game $t$ and
$\Pr\left(\theta_{u}^{t+1} \mid \theta_u^t, h^{t+1}\right)$ is the conditional probability transition between the types. The beliefs computed with \eqref{eq:belief2} are used as the prior beliefs in \eqref{eq:belief1} at the next stage game.

We assume that types are independent and that the type of an agent does not change within a stage game. Therefore, $\mu(\theta \mid h^t)$, which is the probability of a joint type given history of the game, is 
\begin{equation}\label{eq:theta}
\mu(\theta \mid h^t)=\prod_{u \in \mathcal{L}_{\mathfrak{a}}^*(k)}\mu_{u}(\theta_u \mid h^{t}).
\end{equation}

The payoff of agent $v$ after a joint action that transitions $k$-local state $\lambda$ to $\lambda'$ is
\[
\label{x4_eq}
u_v(\lambda,\lambda',\theta_v)=g_v(\lambda',\theta_v)-g_v(\lambda,\theta_v),
\]
where $\theta_v$, which is known to $v$, is considered only if $v\in \mathcal{L}_{\mathfrak{a}}^*(k)$.

Each agent adopts a behavior strategy
to execute an action
at each stage game. Behavior strategy $\sigma_{v}$ assigns a conditional probability over 
actions where the conditioning is on the history of the game, i.e, $\sigma_{v}=\Pr\left(a_{v} \mid h^{t}\right)$.

Next, we leverage Harsanyi-Bellman Ad Hoc Coordination (HBA) \cite{albrecht2013game} to combine Bayesian Nash equilibrium (BNE) with Bellman optimality equation \cite{bellman1957dynamic}. A BNE in round $t$ is a strategy profile $(\sigma_1,\dots,\sigma_{|\mathcal{L}_{\mathfrak{a}}(k)|})$ that maximizes the expected payoffs for all $v\in\mathcal{L}_{\mathfrak{a}}(k)$ with respect to the joint type and available actions in that round. However, this only considers immediate payoffs whereas optimal behavior requires an agent to take payoffs of future rounds into account. Therefore, we combine BNE with Bellman equation to select actions with the maximum expected payoffs, taking into account future rounds.

Given an agent $z$, history $h^t$, and discount factor $\delta$, we use the following operations to get behavior strategy  $\sigma_z(a_{z}|h^t)$.

We compute posterior probability $\mu\left(\theta \mid h^{t}\right)$ using \eqref{eq:belief1}-\eqref{eq:theta}. Using the posterior probability, HBA chooses an action $a_{z}\in A_{z}$ that maximizes the expected payoff $E_{\lambda^{t}}^{a_{z}}\left(h^{t}\right)$ defined as
$$
\begin{gathered}
E_{\lambda}^{a_{z}}(h^t)=\sum_{\theta \in \Theta} \mu\left(\theta \mid h^t\right) E_{\lambda,\theta}^{a_{z}}(h^t)
\end{gathered}
$$
such that $E_{\lambda,\theta}^{a_{z}}(.)$ recursively computes the expected payoff by predicting the future trajectories as
$$
\small{
\begin{gathered}
E_{\lambda,\theta}^{a_{z}}(\hat{h})=\sum_{\vec{a} \in a_{z}\prod\limits_{v \neq z} A_{v}} Q_{\lambda}^{\vec{a}}(\hat{h},\theta) \prod_{v \neq z} \sigma_{v}\left(a_v|\hat{h},\theta_v\right) \\
Q_{\lambda}^{\vec{a}}(\hat{h},\theta)=\sum_{\lambda^{\prime} \in \Lambda} \mathcal{P}\left(\lambda^{\prime}\mid\lambda,\vec{a}\right)\left[u_{z}(\lambda, \lambda',\theta_z)+\delta \max _{a_{z}} E_{\lambda^{\prime},\theta}^{a_{z}}\left(\left\langle\hat{h}, \vec{a}, \lambda^{\prime}\right\rangle\right)\right]
\end{gathered}}
$$
where $\hat{h}$ is the projected history, the purpose of which is to generate future trajectories.
	Note that $\theta_v$ is defined only for $v\in \mathcal{L}_{\mathfrak{a}}^*(k)$, and it hence is not considered for $v\not \in \mathcal{L}_{\mathfrak{a}}^*(k)$ although it is written in the computations. Moreover, we assume that the depth of the planning horizon (recursive calls) is infinite, which necessitates discount factor $\delta$, nevertheless the system is self-stabilizing. However, one may set $\delta=1$, because all trajectories end to the 
	configurations where hereafter for all $v\in\mathcal{L}_{\mathfrak{a}}(k)$, $|A_{v}^{\lambda, \theta}|=1$ (i.e., a termination condition). Afterwards, one may also use \eqref{eq:belief2} and \eqref{eq:theta} to update the projected beliefs for a limited number of the trajectories due to
	\red{unstabilized states}
	in the initial rounds.

Finally, a strategy profile $(\sigma_1,\dots,\sigma_{|\mathcal{L}_{\mathfrak{a}}(k)|})$ is optimal, a BNE, in round $t$ if it simultaneously maximizes the expected payoff for all $z\in \mathcal{L}_{\mathfrak{a}}(k)$ in round $t$. Note that for each $z\in \mathcal{L}_{\mathfrak{a}}(k)$, a BNE places positive probabilities only on actions $a_{z} \in A_{z}$ that maximize the expected payoff, i.e., $\arg\max_{a_{z}} E_{\lambda^{t}}^{a_{z}}\left(h^{t}\right)$.

Given \eqref{eq:stabtprim} and the fact that the system stabilizes, there exists a round after which HBA knows the agent types and, since 
it always learns the same from a given history (all types are deterministic learners\red{, i.e., given history and model parameters the resulting type is deterministic}), the expected payoffs are correct \cite{albrecht2013game}. Besides since the type distribution is always absolutely continuous (\eqref{eq:belief1} and \eqref{eq:belief2}), according Theorems 1 and 2 in \cite{kalai1993rational}, the system converges to a Nash equilibrium.

\subsection{Perturbation Avoidance after Closure}
\label{subsec:approach:perturbation}

\red{As we discussed in \Cref{section:problem},} a perturbation is as a unilateral and specified change (e.g., $True \Rightarrow False$) of one primary variable of an agent in a legitimate configuration that turns that legitimate configuration into an illegitimate one. \red{We define a $k$-fault as changing $k$ primary variables in a legitimate configuration of the system by arbitrarily transient faults.} Based on this definition, a perturbation is equivalent to a 1-fault \cite{ghosh1996fault}. \red{Henceforth, we use the term 1-fault to refer to perturbations.} The configuration  derived from a 1-fault differs from a legitimate one only in the variables of the faulty agent. Works on self-stabilization \cite{ghosh1996fault, ghosh_fault-containing_2007} 
show that when an agent is in a legitimate configuration and a 1-fault occurs: a) that agent becomes enabled and b) it reaches a stable configuration by the agent executing one of its enabled rules. We aim to detect and resolve 1-faults by activating rules only in the faulty agent to prevent other agents from experiencing the fault.

\red{Formally, we define a 1-fault as a tuple $\left<v,X,\left(x_1,x_2\right)\right>$ (denoted shortly by $\hat{e}_v$), where $v\in V$ is the faulty agent, $X\in Vars$ is the corrupted variable, and $x_1,x_2\in\mathbb{V}_X$ are the values of $X$ before and after the change point, respectively.} 

\begin{definition}\label{d5}{\em Depth of Contamination.} Let $R_v$ be a subgraph induced by the agents involved in recovering from a 1-fault
$\hat{e}_v$.
Then, the depth of contamination denoted by $D(\hat{e}_v)$
is the distance from $v$ to the farthest agent in $R_v$, i.e., $D(\hat{e}_v) = \max\{dist(v,w)\mid w\in R_v\}$.
\end{definition}

\begin{definition}\label{d4}{\em Fault-Containment.}
	Let $Q$ be a self-stabilizing DIS and let $\hat{e}_v$ denote a 1-fault.
	$Q$ is fault-containing for $\hat{e}_v$ if $D(\hat{e}_v)$ remains constant regardless of the number of agents in the system.
\end{definition}

\begin{theorem}\label{t2} If a self-stabilizing system always returns to the legitimate configuration that the system was in before a perturbation, that legitimate configuration is a Nash equilibrium. 
\end{theorem}

\begin{proof}According to the definition of Nash equilibrium in self-stabilization, a legitimate configuration is a Nash equilibrium if no agent can profit from a unilateral selfish action. The incentive for an agent to perturb the system is that convergence to another legitimate configuration allows it to gain more profit in the new configuration than it could in previous ones. If after a perturbation, a self-stabilizing algorithm causes the system to again converge to the last legitimate configuration, no agent will have incentive to cause a perturbation, and thus that configuration is a Nash equilibrium. \end{proof}

\begin{definition}\label{d3}{\em Perturbation-Proof.}
	A self-stabilizing system is perturbation-proof if all legitimate configurations are Nash equilibria for the set of gain functions associated with the agents.
\end{definition}

\begin{theorem}\label{t4} Let a self-stabilizing DIS $Q$ contain 1-fault $\hat{e}_v$ with contamination depth zero. Then, $Q$ is perturbation-proof for the set of gain functions that can lead to $\hat{e}_v$.\end{theorem}

\begin{proof} After $\hat{e}$ occurs
, since $D(\hat{e}_v)$ is zero, only $v$ will change its variables; otherwise, the system remains in an illegitimate configuration. By doing so, the system converges to the last legitimate configuration through the execution of a finite number of rules by $v$. Therefore, according to \Cref{t2}, the legitimate configurations of $Q$ are Nash equilibria for any gain function that leads to $\hat{e}_v$.
\end{proof}

\subsection{Self-Stabilization facing Selfishness: a 3-Step Approach}
\label{subsec:approach:deflection}

Given that no illegitimate configuration is a Nash equilibrium, one can design a self-stabilizing algorithm that guarantees both convergence and closure properties even when selfish agents are able to deflect (i.e., violate, execute an unauthorized action during convergence, or perturb), using the following sequence of three operations:
\begin{enumerate}[i.]
	\item First, we identify unauthorized actions and their corresponding guards and resolve them by adding rules to the system. \red{We call these rules \textit{selfish rules}}.
	\item Then, we modify the system to be perturbation-proof. 
	\item Last, we randomize the rules with behavior strategies. Doing so not only make the system
	tolerate
	violations but also confirms the closure property because the probabilities of executing \red{selfish rules} become zero in a legitimate configuration.
\end{enumerate}

According to \Cref{t9}, the final solution is weak-stabilizing \cite{gouda2001theory}. \red{Under a weak-stabilizing algorithm, from any arbitrary configuration, there is a set of computations that eventually reaches a legitimate configuration.}

\begin{theorem}\label{t9}
	Assume a system where agents can arbitrarily execute actions that do not comply with their enabled rules. Then, a perturbation-proof and violation-tolerant self-stabilizing solution to the system is weak-stabilizing. 
\end{theorem}

\begin{proof}
	Consider any unauthorized action as a new rule. Because the order of evaluations of the rules is arbitrary, we assume a sequence of computations where the distributed scheduler 
	\red{never select a selfish rule.}
	 Because the system is violation-tolerant, the agents will not permanently violate the rules and hence the system reaches a legitimate configuration (possible convergence).
	Then, since the system is perturbation-proof, no agent has incentive to execute an unauthorized action (closure).\end{proof}

Noting that any weak-stabilizing algorithm is self-stabilizing under a distributed randomized scheduler \cite{devismes2008weak}, the third step (i.e., randomization with behavior strategies) significantly reduces the convergence time of the algorithm because under the assumption that
no illegitimate configuration is a Nash equilibrium,
the probabilities of \red{selfish} rules are never one and their mean converges to zero as the number of enabled agents and thus the competition decreases.

\section{A Non-cooperative Self-Stabilizing Approach to Clustering}
\label{section:casestudy}

In this section, we illustrate earlier concepts through a case study: clustering in a DIS. We present algorithms that provably converge and exhibit closure despite the presence of selfish agents.

Clustering is an important technique for achieving scalability in randomly deployed DISs. We consider the case of a connected failure-prone network of autonomous 
energy-constrained 
agents. A self-stabilizing algorithm for clustering provides automatic recovery from transient faults caused by unexpected failures, a dynamic environment, or topological changes and needs no specific initialization.

In light of energy-constraints, a clustering solution should satisfy two properties. First, to allow efficient communication between each pair of agents, every agent should have at least one cluster-head in its neighborhood. Second, it is desirable to prevent cluster-heads from being neighbors in order to improve energy efficiency and network throughput \cite{moscibroda2004efficient}. These two properties lead to the concept of a maximal independent set (MIS) in graph theory.

We model a DIS as a graph $G = (V,E)$. An {\em independent set} $(IS) $ of $G$ is a subset of $V$ such that $\forall v,u\in IS: (v,u)\notin E$. Members of $IS$ are referred to as {\em cluster-heads}. $IS$ is an {\em MIS} if any agent $v \notin IS$ has a neighbor in $IS$.

We assume a message passing computational model with communication topology $G$. While an agent can only change its own state, it is able to directly read the states of agents within distance two from it (e.g, using full transmission power to send states and half of it to communicate data).


In \cite{shukla1995observations}, the authors proposed the first self-stabilizing algorithm to construct an MIS (\Cref{fig:basicmis}). This algorithm works with either a central scheduler or a distributed randomized scheduler.
\begin{algorithm}[ht]
	\begin{algorithmic}
		\Variables$\ \ state$: Binary
		\EndVariables
		\Predicates
		\State $\begin{aligned}[t]
		& pending(v) \equiv v.state=\mathrm{OUT} \wedge \forall w \in N(v):w.state=\mathrm{OUT}\\
		& conflict(v) \equiv v.state=\mathrm{IN} \wedge \exists w \in N(v):w.state=\mathrm{IN}
		\end{aligned}$
		\EndPredicates
		\Rules
		\State $\begin{aligned}[t]
		{R_1}:\ \ \ & pending(v) \longrightarrow v.state:=\mathrm{IN}\\
		{R_2}: \ \ \ & conflict(v) \longrightarrow v.state:=\mathrm{OUT} 
		\end{aligned}$
		\EndRules
	\end{algorithmic}
	\captionof{algocf}{Basic MIS ($b$MIS)}
	\label{fig:basicmis}
\end{algorithm}

The primary variable $state$ is a binary variable that specifies whether agent $v$ belongs to an independent set (IS) or not. If $v$ and its neighbors are not in an IS ($pending$), the guard of the first rule evaluates to true, and thereupon $v$ becomes a member of the IS. Furthermore, if $v$ and one of its neighbors both belong to an IS ($conflict$), the guard of the second rule evaluates to false and thereupon $v$ exits from the IS.


Here, the system property is \eqref{1_eq}, which evaluates to true if there is no $conflict$ or $pending$ in the system, 
\begin{equation}
\label{1_eq}
\small{
	\begin{aligned}
	\mathbb{P}\equiv\forall v:\big(\neg conflict(v) \wedge \neg pending(v)\big)
	\end{aligned}}.
\end{equation}

We assume an agent gains a reward $\vartheta>0$ from being part of a cluster
and incurs a cost $\zeta>0$ from being a cluster-head (communication and computation overhead), where $\zeta<\vartheta$ \cite{koltsidas_game_2011}. In this case, cluster members obtain a profit $\vartheta$ and cluster-heads a profit $\vartheta-\zeta$. We define the gain as
\begin{equation}
\label{2_eq}
g_v=\begin{cases}
0, &\textrm{if $\ pending(v)$,}\\ 
\vartheta,  &\textrm{if  $\ v.state=\mathrm{OUT} \wedge \neg pending(v)$,}\\ 
\vartheta-\zeta,  &\textrm{if  $\ v.state=\mathrm{IN}$.} \end{cases}
\end{equation}

\begin{theorem}\label{t11}
	Assume a system $Q$ with payoff function \eqref{2_eq} forms an MIS. Then, no non-MIS configuration of $Q$ is a Nash equilibrium. 
\end{theorem}

\begin{proof}Assume there is a non-MIS configuration $c$ that is a Nash equilibrium. 
Because $c$ is not an MIS, there exists an agent $v$ that either is pending or has a conflict with a neighboring agent $w$. In the first case, $v$ can increase its profit from zero to $\vartheta-\zeta$ by making a move, and in the second case, because the system reaches an MIS, one of $v$ and $w$ will not be a cluster-head at the end and thus it will profit $\zeta$ if it makes a move. This contradict the assumption that $c$ is a Nash equilibrium.\end{proof}

The underlying assumption in previous works on self-stabilizing construction of an MIS \cite{guellati_survey_2010} is that the rules are followed by all agents. However, in the light of non-cooperative DISs, if agents are selfish, previously proposed algorithms will not function as expected.
Therefore, we need to redesign algorithms so that a selfish agent will conform to the modified rules.

\subsection{Towards a Violation-Tolerant Solution}
\label{subsection:vtmis}

Energy-constrained selfish agents may be reluctant to join an independent set (IS) and hence may hesitate to follow the rules. An example where 
two agents decide to violate a rule is illustrated in Figure \ref{fig:example2}.
Such behavior can be modeled as a non-cooperative game.

\begin{figure}[h]
\centering
\includegraphics[width=0.35\linewidth]{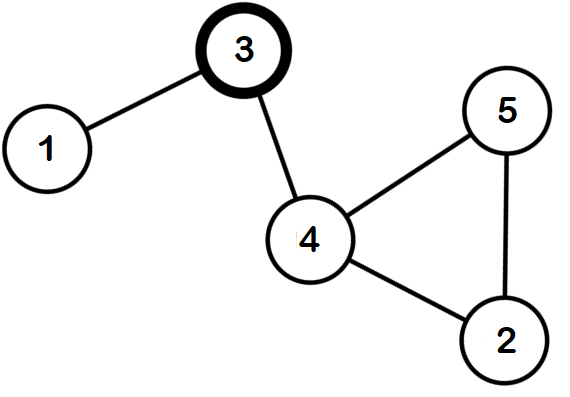}
\caption{Thicker nodes denote cluster-heads. Neither agent 2 nor 5 prefers being a cluster-head, so they may violate any rule that makes them enter the IS.}
\label{fig:example2}
\end{figure}

We model self-stabilizing clustering of selfish agents (SoS) as an extensive form game consisting of a sequence of stage games played in every round. In SoS, each player has at most two possible actions, Switch or Preserve, corresponding, respectively, to the options of changing the value of $state$ or maintaining it as is. Expected payoffs depend on the topology, configuration, algorithm, and gain function. 

We aim to design a probabilistic self-stabilizing algorithm that tolerates deviations. 
To do so, we consider $b$MIS and make its first rule probabilistic. Then, we transform SoS to a stochastic Bayesian game
and solve it to obtain the behavior strategies required to generate probabilistic rules. Doing this yields a violation tolerant algorithm (\Cref{fig:drmis}) that ensures a selfish agent enters the IS with probability $p>0$, i.e., the behavior strategy that corresponds to action Switch. Note that $p$ varies depending on both agent and configuration.
\begin{algorithm}[ht]
	\begin{algorithmic}
	\Variables
\State $state$: Binary
\State $p$: Real \tcp{$p\in(0,1]$}
\EndVariables
	\Predicates
	\State $\begin{aligned}[t]
	& pending(v) \equiv v.state=\mathrm{OUT} \wedge \forall w \in N(v):w.state=\mathrm{OUT}\\
	& conflict(v) \equiv v.state=\mathrm{IN} \wedge \exists w \in N(v):w.state=\mathrm{IN}
	\end{aligned}$
	\EndPredicates
	\Rules
	\State $\begin{aligned}[t]
	{R_1}:\ \ \ & pending(v) \xrightarrow{\text{\ \ $p$\ \ }} v.state:=\mathrm{IN}\\
	{R_2}: \ \ \ & conflict(v) \longrightarrow v.state:=\mathrm{OUT} 
	\end{aligned}$
	\EndRules
\end{algorithmic}
\captionof{algocf}{violation tolerant MIS ($vt$MIS)}
\label{fig:drmis}
\end{algorithm}

\begin{theorem}\label{t145}
 Under the assumption that the gain function is given by \eqref{2_eq}, $vt$MIS is probabilistically self-stabilizing under a distributed fair scheduler.
\end{theorem}

\begin{proof}The state of every agent enabled by $R_2$ transforms to OUT and remains unchanged unless $R_1$ is enabled. If an agent $v$ is enabled by $R_1$, it enters the IS with probability $p>0$. 
Noting that $p$ is less than one if $R_1$ is enabled by $w \in N(v)$, there is a probability $q>0$ that among $v$ and its neighbors, during a round, one agent, e.g., $u \in \{v\} \cup N(v)$, enters the IS while no neighbor of $u$ (including $v$ if $u \neq v$) takes an action. After this case occurs, $u$ and its neighbors will no longer be enabled, i.e., the state of $v$ remains unchanged. Therefore, the probability that $v$ will not be enabled after $k$ rounds is $1-(1-q)^k$, which converges to one as $k\rightarrow \infty$, where $k$ is the number of completed rounds. 
Similar consecutive executions occur at other agents until an MIS is formed. \red{Strictly speaking the probability that an MIS will eventually be found converges to one (probabilistic convergence) and then there will be no enabled rule in the system (closure)}; therefore, $Q$ is a probabilistic self-stabilizing system.
\end{proof}

According to \eqref{2_eq}, the gain of each agent in each configuration depends on its state and the states of its neighbors. In each configuration, some agents may take action but the new states are not realized until the next round. In this regard, each agent needs to know the gain functions of its neighbors, which depend on their neighbors, to make the move that is best for it given the moves of its neighbors. This dependency chain results in a network game. We assume that the knowledge of any agent about the network graph is limited to its two-hop distance.
Consequently, each agent is playing with its one-hop and two-hop neighbors; however, it needs to estimate the payoff of agents about whom it has incomplete information, its two-hop neighbors, denoted by $\mathcal{L}_v^*(2)$. To do so, it creates beliefs about the types of its two-hop neighbors and updates these beliefs at the end of each round.

We define a joint type $\theta=\prod_{u \in \mathcal{L}_{v}^*(k)}\theta_u$ as a specification on the states of agents beyond the boundaries of neighborhood, where $\theta_u$ is a binary variable that indicates whether an agent $u \in \mathcal{L}_{\mathfrak{a}}^*(k)$ has a neighbor that is cluster-head or not.

In each round, given that communications is fast and reliable, under a distributed scheduler, a new observation resembles a $2$-local state that may differ from the previous one only in the variables of agents that were selected by the scheduler to execute actions in the previous round. 
In other words, the next $2$-local state is one of the possible $2$-local states in the path of transition from the current one to another one that is realized after taking all the given actions.



\red{We consider a simplified model of a distributed randomized scheduler.
Let $X_{v}^t$ be a Bernoulli variable that takes value 1 if the scheduler selects agent $v$ at round $t$ and 0 otherwise.
We assume $X_{v}^t$ for all $v\in V$ and $t\in \mathbb{N}^0$ follows the same distribution $$X_{v}^t \stackrel{\text { iid }}{\sim} \operatorname{Bernoulli}(p_S).$$}
Let $\nabla$ denote the set of primary variables that differ in values between $2$-local states $\lambda$ and $\lambda'$, and let $\partial$ be the number of all variables whose values \red{can potentially} change after joint action $\vec{a}$ in $\lambda$. We have
\begin{equation}
\label{eq:dis.transition}
\small{\mathcal{P}(\lambda'\mid \lambda,\vec{a})=
\begin{cases}
p_S^{|\nabla|}(1-p_S)^{\partial-|\nabla|} & \mbox{if $\forall \mathrm{var}\in\nabla:\lambda.\mathrm{var}\xrightarrow{\vec{a}}\lambda'.\mathrm{var}$}\\ 
0 &\mbox{otherwise } \end{cases}}
\end{equation}



\red{We model the behavior of selfish agents as a stochastic Bayesian game (\Cref{subsec:gamemodel}) with respect to a self-stabilizing clustering problem (i.e., SoS) where agents run $vt$MIS and are able to violate the rules.}

Assume agent $\mathfrak{a}\in V$ is pending.
Hence, $\mathfrak{a}$ decides to execute the first rule or violate it according to its behavior strategy profile given an equilibrium exists in the game. 

We define the type of agent $u \in \mathcal{L}_{\mathfrak{a}}^*(2)$ as
\begin{equation}
\label{typeeq}
\theta_u =\begin{cases}
\mathrm{IN}, &\textrm{if $\exists w \in N(u)/\mathcal{L}_{\mathfrak{a}}(1): w.state = \mathrm{IN}$,}\\ 
\mathrm{OUT}, &\textrm{if $\forall w \in N(u)/\mathcal{L}_{\mathfrak{a}}(1): w.state = \mathrm{OUT}$.} \end{cases}
\end{equation} 

Given $2$-local state $\lambda\in \{IN,OUT\}^{|\mathcal{L}_{\mathfrak{a}}(2)|}$ and joint type $\theta \in \{IN,OUT\}^{|\mathcal{L}_{\mathfrak{a}}^*(2)|}$, where $\Gamma=\{IN, OUT\}$, the set of available actions for agent $v\in V$ is
$$A_v{(\lambda,\theta)}=\begin{cases}
\{\textrm{Switch, Preserve}\} &\textrm{\small{if $pending(v)$ given $\lambda$ and $\theta$}}\\ 
\{\textrm{Switch}\} &\textrm{\small{if $conflict(v)$ given $\lambda$ and $\theta$}} \\ 
\{\textrm{Preserve}\}  &\textrm{otherwise} \end{cases}$$

\red{At the first round, agent $\mathfrak{a}$ estimates that the initial state of every agent is OUT with probability $p_0=|\{w\in\mathcal{L}_\mathfrak{a}(2)|w.state=\mathrm{OUT}\}|/|\mathcal{L}|$. Then, with respect to \eqref{typeeq}, for all $u\in \mathcal{L}_\mathfrak{a}^*(2)$, it approximates $\mu(\theta_u=\mathrm{OUT}\mid h^0)=p_0^{\max(0,\bar{d}-|N_{\cap}(\mathfrak{a},u)|)}$, where $\bar{d}=\lceil \sum_{w\in \mathcal{L}_\mathfrak{a}(2)} \mathrm{degree}(w)/|\mathcal{L}_\mathfrak{a}(2)|\rceil$ and $N_{\cap}(\mathfrak{a},u)=\{w\in \mathcal{L}_{\mathfrak{a}}(1)|u\in N(w)\}$.}
	
\red{We update the beliefs regarding the next stage only if $$\sum_{\theta_u\in\Theta_u}\Pr(\theta_u|\theta_u,h^{t+1})\mu(\theta_u|h^{t+1})<0.5;$$otherwise, we assume the neighbors of $u$ are stabilized.} \red{To approximate $\Pr(\theta_u^{t+1}\mid \theta_u^{t},h^{t})$, we do as follows. Let function $\operatorname{state}(v,t)$ returns the state of agent $v$ at a given round $t$. We define the type of neighborhood of an agent $v$ at rounds $t$ as 
$$\footnotesize{\operatorname{type}(v,\lambda)=
\begin{cases}
\mathrm{IN}, &\theta^t=IN \vee \exists w\in N(a):v\in N(w) \wedge\operatorname{state}(w,\lambda)=\mathrm{IN},\\
\mathrm{OUT}, &\text{otherwise}.
\end{cases}}$$
From the history of the game, we only consider the state of $u$ at rounds $t$ and $t+1$, denoted by $\operatorname{st}_0$ and $\operatorname{st}_1$, respectively, and the type of its neighborhood at rounds $t$ and $t+1$, denoted by $\operatorname{tp}_0$ and $\operatorname{tp}_1$. Next, we consider a $\bar{d}$-regular graph $G'(V',E')$, where $|V'|=|\mathcal{L}_v(2)|$, as a network of agents. We assume the initial states of agents comes from Bernoulli distribution $\Pr(\mathrm{IN})=1-p_0$. Assume $c^0$ as an initial configuration. We model the stabilization problem using a stochastic game of complete information where agents run $vt$MIS, know each other's state, and take actions analogous to subgame-perfect equilibrium. Solving the game, we have the probability distribution of next actions for every agent $v$ as $\Pr(\vec{a}|c^0)$. Then, using those probability distributions, we compute the distribution of the next configuration as $\Pr(c^1|c^0)=\Pr(\vec{a}|c^0)\Pr(c^1|c^0,\vec{a})$, where $\Pr(c^1|c^0,\vec{a})$ depends to the distributed scheduler and is computed according to \eqref{eq:dis.transition}. 
Let function $\operatorname{st}'(v,c)$ returns the state of agent $v\in V'$ in configuration $c$. We define function
\[\operatorname{tp}'(v,c)=
\begin{cases}
\mathrm{IN}, &\exists w\in N(v):\operatorname{st}'(w,c)=\mathrm{IN},\\
\mathrm{OUT}, &\forall w\in N(v):\operatorname{st}'(w,c)=\mathrm{OUT},
\end{cases}
\]
which returns the exact type of neighborhood of agent $v\in V'$ in configuration $c$ for graph $G'$. Then, we have
$$
\footnotesize{\begin{aligned}
&\Pr(\theta_u^{t+1}\mid \theta_u^{t},h^{t})=\\&\frac{\sum\limits_{c_0,c_1\in C}\sum\limits_{v\in V'}\Pr(c_0)\Pr(c_1|c_0)\mathbf{1}_{\left[\operatorname{tp}'(v,c_0)=\operatorname{tp}_0 \wedge \operatorname{tp}'(v,c_1)=\operatorname{tp}_1 \wedge \operatorname{st}'(v,c_0)=\mathrm{st}_0 \wedge \operatorname{st}'(v,c_1)=\mathrm{st}_1\right]}}{\sum\limits_{c_0,c_1\in C}\sum\limits_{v\in V'}\Pr(c_0)\Pr(c_1|c_0)\mathbf{1}_{\left[\operatorname{tp}'(v,c_0)=\operatorname{tp}_0 \wedge \operatorname{st}'(v,c_0)=\mathrm{st}_0 \wedge \operatorname{st}'(v,c_1)=\mathrm{st}_1\right]}},
\end{aligned}}
$$
where $\mathbf{1}$ is the indicator function.}



\subsection{Towards a Perturbation-Proof Solution}

According to \Cref{d3}, neither $b$MIS nor $vt$MIS is perturbation-proof with respect to \eqref{2_eq} because a selfish agent can profit from perturbing a legitimate configuration. Figure \ref{fig:example1} shows an example where a selfish agent may benefit from perturbing a legitimate configuration.

\begin{figure}[h]
	\centering
	\subfloat[configuraion I]{\includegraphics[width=0.35\linewidth]{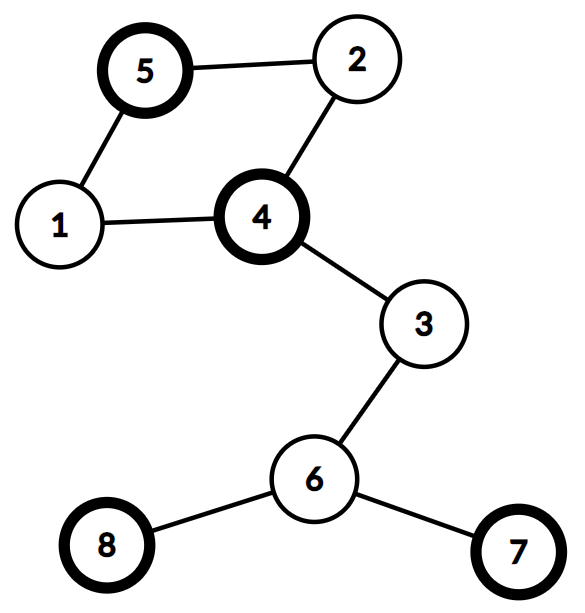}
	}
	\subfloat[configuraion II]{\includegraphics[width=0.35\linewidth]{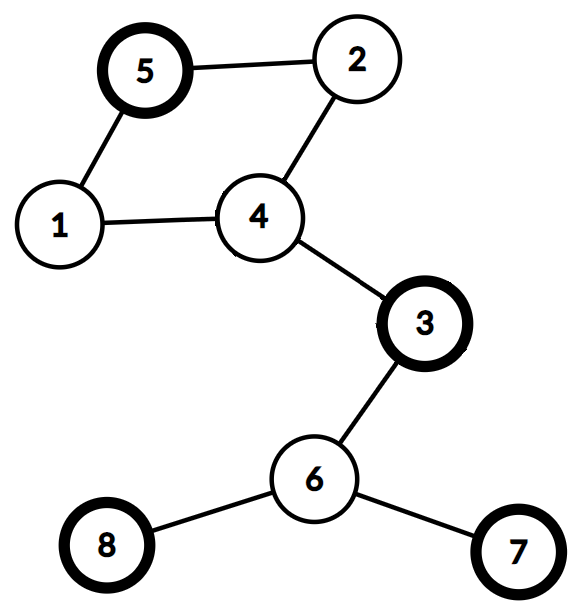}}
	\caption{Given that agent 4 perturbs the legitimate configuration I, it may benefit if the system converges to the legitimate configuration II.}
	\label{fig:example1}
\end{figure}

We propose $pf$MIS (Algorithm \ref{fig:pfmis}) as a perturbation-proof solution for MIS ($pf$MIS) to handle any IN-to-OUT 1-fault . 


The first two rules of $pf$MIS are similar to their equivalents in $bMIS$ with the introduction of the predicate $hesitate$ to the first rule. 
The secondary variable $parent$
contains the identity of the
neighboring cluster-head of agent $v$; in case the agent has two or more neighboring cluster-heads, it takes value $\perp$. Rules $R_3$-$R_5$ are responsible for setting the value of $parent$. \red{Note that the symbol "$\exists!$" in $R_3$ denotes the unique existential quantification.}

The predicate $hesitate$ allows a pending agent $v$ to determine whether any of its neighbors has incurred a 1-fault: if $v.parent$ corresponds to one of the neighbors of $v$, such as $w$, $u.parent$ in each neighbor $u$ of $v$ except $w$ (i.e., $u\in N(v)/\{w\}$) is not $v$, and the states of all neighbors of $w$ are OUT and $parent$ in all of them is $w$ or $\perp$, $v$ concludes that an IN-to-OUT 1-fault  has occurred in $w$ except in the special case where $w.parent=v$ and $z.parent$ in each neighbor $z$ of $w$ except $v$ (i.e., $z\in N(w)/\{v\}$) is $\perp$. In that case, because the fault may have occurred in either $v$ or $w$, the agent with the smaller $id$ has to enter the IS. In this paper, to keep the contamination depth to zero, we assume that an agent can assign another agent to be its $parent$ only if the other agent is a cluster-head, so that the special case, where $id$s are compared, does not happen.

\begin{algorithm}[ht]
	\begin{algorithmic}
		\Constants
		\State $id$: Integer
		\EndConstants
		\Variables
		\State $state$: Binary
		\State $parent$: Agent 
		\EndVariables
		\Predicates
		\State $\begin{aligned}[t]
		& pending(v) \equiv v.state=\mathrm{OUT} \wedge \forall w \in N(v):w.state=\mathrm{OUT}\\
		& conflict(v) \equiv v.state=\mathrm{IN} \wedge \exists w \in N(v):w.state=\mathrm{IN}
		\end{aligned}$
		\State $\begin{aligned}[t]
		hesi&tate(v) \equiv v.parent \not = \perp \wedge \exists w \in N(v): \big[v.parent=w \\&\wedge\forall z\in N(v)\setminus\{w\}:z.parent\not =v \wedge \forall u \in N(w) \setminus \{v\}:\\&\big(u.state=\mathrm{OUT} \wedge (u.parent = w \vee u.parent = \perp)\big) \\&\wedge
		(w.parent=\perp \vee w.parent \not= v \vee w.id<v.id \\&\vee \exists y\in N(w) \setminus \{v\}:y.parent=w)\big]
		\end{aligned}$
		\EndPredicates
		\Rules
		\State $\begin{aligned}[t]
		{R_1}:\ \ \ & pending(v)\wedge \neg hesitate(v) \longrightarrow v.state:=\mathrm{IN}\\
		{R_2}:\ \ \  & conflict(v) \longrightarrow v.state:=\mathrm{OUT}\\
		{R_3}:\ \ \  &v.state=\mathrm{OUT} \wedge \not\exists u \in N(v):pending(u) \wedge \\&\ \ \ \exists! w \in N(v): (w.state = \mathrm{IN} \wedge v.parent \not = w)\\&\ \ \ \longrightarrow v.parent:= w\\
		{R_4}:\ \ \  &v.state=\mathrm{OUT} \wedge \exists w,z \in N(v): (w \not =z \wedge w.state = \mathrm{IN} \\& \ \ \ \wedge z.state = \mathrm{IN}) \wedge v.parent \not = \perp \longrightarrow v.parent:= \perp\\
		{R_5}:\ \ \  &v.state=\mathrm{IN} \wedge \neg conflict(v)\wedge v.parent \not = \perp \\& \ \ \ \longrightarrow v.parent:= \perp
		\end{aligned}$
		\EndRules
	\end{algorithmic}
	\captionof{algocf}{perturbation-proof MIS ($pf$MIS)}
	\label{fig:pfmis}
\end{algorithm}

\red{We begin by proving that no pending agent deadlocks when $hesitate(\cdot)$ is true.}

\begin{lemma}\label{le00} Suppose both $pending(v)$ and $hesitate(v)$ are true for an agent $v$ in a system running $pf$MIS. Then, $v$ has at least one neighbor that has $R_1$ enabled.\end{lemma}
\begin{proof}Because $hesitate(v)$ is true, there exists $w \in N(v)$ such that $v.parent=w$ and $pending(w)$ is true. We prove that $hesitate(w)$ is false by contradiction and thus $R_1$ is enabled for $w$. Suppose $hesitate(w)$ is true. Then, $u.parent$ in every neighbor $u$ of $w$ except $w.parent$ cannot be $w$, i.e., $hesitate(w) \Rightarrow \forall u\in N(w)/\{w.parent\}: u.parent \not = w$. Because $v.parent=w$, we conclude $w.parent=v$. In this case, $hesitate(w)$ is true only if $w.id>v.id$ or $v$ has a neighbor $z$ except $w$ that $z.parent=v$. In both conditions, $hesitate(v)$ is false which contradicts the assumption. \end{proof}

\begin{lemma}\label{le0}If no agent is enabled in a system running $pf$MIS, the configuration is an MIS.\end{lemma}
\begin{proof}Suppose that the system is stabilized and its configuration denoted by $c$ is not an MIS. This leads to two cases. In the first case, $c$ is not an IS (independent set). So, there exists at least two neighbors $v$ and $w$ that are IN. In this case, $R_2$ is enabled at either $v$ or $w$, contradicting our assumption. In the second case, $c$ is an IS but is not maximal. In this case, there must be at least one OUT agent $v$ that has no IN neighbor, which means $pending(v)$ is enabled. In this case, $R_1$ is not enabled for $v$ only if $hesitate(v)$ is true. Therefore, according to \Cref{le00}, $v$ has at least one neighbor $w$ that is enabled by $R_1$, which is a contradiction.\end{proof}

\begin{lemma}\label{le1}(Convergence) Starting with any configuration, a system running $pf$MIS eventually reaches an MIS under a distributed randomized scheduler.\end{lemma}
\begin{proof}In an initial configuration, the state of an agent $v$ can be IN or OUT. In the case of IN, if $v$ has no IN neighbor (i.e., no conflict situation), $v$ and its neighbors permanently maintain their states; otherwise, $v$ may execute $R_2$ and become OUT. In the case of OUT, if $v$ has no IN neighbor and $hesitate(v)$ is true, it waits for one of its neighbors to become IN (\Cref{le00}) after which it cannot execute any rule, or if $hesitate(v)$ is false, it executes $R_1$ and its state becomes IN. Since the scheduler is randomized, there is a probability $p>0$ that among the agents in the neighborhood of $v$, only $v$ executes $R_1$ and permanently becomes IN. Finally, if $v$ has an IN neighbor $w$, it maintains its OUT state as long as $w$ is IN. When no agent is enabled, an MIS is created and remains so according to \Cref{le0}.\end{proof}

\begin{lemma}\label{le2}(Closure) Given that a system running $pf$MIS is in a legitimate configuration, it will not leave it, provided no fault occurs.\end{lemma}
\begin{proof} We prove the closure property by contradiction. Suppose the system is in a legitimate configuration for which the closure condition does hold. This means at least one rule that changes the primary variable $state$ is enabled. Suppose the enabled rule is $R_1$ or $R_2$. Then, the configuration is not an MIS which contradicts the assumption that the system is in a legitimate configuration. Note that once an MIS is created, Rules 3-5 are enabled at most once after which no rule is enabled.\end{proof}

\begin{theorem}\label{t13}
	$pf$MIS is self-stabilizing under a distributed randomized scheduler.
\end{theorem}
\begin{proof}According to \Cref{le1,le2}, $pf$MIS exhibits both closure and convergence. It is thus a self-stabilizing algorithm.\end{proof}

\begin{lemma}\label{le4}(Fault-containment) If an IN-to-OUT 1-fault  occurs in a system running $pf$MIS, only the faulty agent changes its state during convergence.\end{lemma}

\begin{proof} Suppose an IN-to-OUT 1-fault  occurs in $v$. Because $v$ was a member of an MIS before the fault, $pending(v)$ becomes true. Then, it suffices to check that $hesitate(v)$ is false while it is true for the neighbors of $v$; therefore, only $v$ can execute $R_1$. Note that we assume $v$ cannot assign $w \in N(v)$ to $parent$ unless $w$ is a cluster-head; otherwise, in the rare case that no other agent except $w$ has assigned $v$ to its $parent$ and $w$ has no other IN neighbor, the term $v.id>w.id$ will be evaluated and if it is true, the contamination depth will be one instead of zero because $w$ enters the IS instead of $v$. With this in mind, with the execution of actions in the faulty agent and only that agent, the system converges to a legitimate configuration.\end{proof}

\begin{theorem}\label{t14}
	DIS $Q$ is perturbation-proof if it runs $pf$MIS using the gain function $\mathrm{(2)}$.
\end{theorem}

\begin{proof} Following (2), the only possible perturbation is an IN-to-OUT 1-fault . As stated in \Cref{le4}, $Q$ recovers from any IN-to-OUT 1-fault  with contamination depth zero. According to Theorem 2, $Q$ is thus perturbation-proof.\end{proof}

\subsection{Towards a Deflection-Tolerant Solution}
\label{subsection:dtmis}

Next, with the help of the three operations introduced in \Cref{subsec:approach:deflection}, we present a self-stabilizing algorithm that tolerates deflections. An example of possible deflections in a system is illustrated in \Cref{fig:example3}.

\begin{figure}[h]
	\centering
	\includegraphics[width=0.35\linewidth]{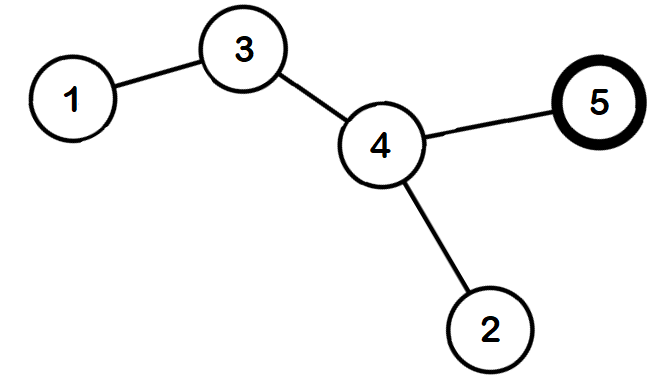}
	\caption{Agents numbered 1 to 3 and 5 may deflect: agents numbered 1 to 3 violate the rules of entry into the IS and number 5 leaves the IS.}
	\label{fig:example3}
\end{figure}

We start by identifying unauthorized actions \red{and their corresponding guards to transform them to selfish rules}. 
There is one selfish rule, namely that an agent exits an IS \red{when there is no conflict}. Then, we design a perturbation-proof algorithm based on $pf$MIS; however, we modify it to make it able to handle any IN-to-OUT 1-fault  despite \red{the new selfish rule}. Finally, we model 
\red{the behavior of agents as a game}
and transform any rule that is violation-prone or \red{selfish} into a probabilistic rule accounting for
behavior strategies.

The resulted algorithm, $dt$MIS (\Cref{fig:smis}), is a clustering algorithm that tolerates deflections and works under a distributed randomized scheduler. \red{The last rule of $dt$MIS ($R_7$) is the added selfish rule that represents the unauthorized action of an agent that leaves the IS when there is no conflict.
The basic idea of the perturbation-proof part of $dt$MIS is similar to that of $vt$MIS. 
The predicate $hesitate$ is refined and a secondary variable $parents$ is employed as a list of neighboring cluster-heads to make the algorithm perturbation-proof despite $R_7$. Note that $R_1$ and $R_7$, which are prone to violation and deflection, respectively, are the only probabilistic rules.}

\begin{algorithm}[ht]
	\begin{algorithmic}
		\Constants
		\State $id$: Integer
		\EndConstants
		\Variables
		\State $state$: Binary
		\State $parents$: Array \tcp{set of agents}
		\State $p$: Real \tcp{$p\in(0,1]$}
		\State $q$: Real \tcp{$q\in[0,1)$}
		\EndVariables
		\Predicates
		\State $\begin{aligned}[t]
		& P(v) \equiv v.parents\\
		& pending(v) \equiv v.state=\mathrm{OUT} \wedge \forall w \in N(v):w.state=\mathrm{OUT}\\
		& conflict(v) \equiv v.state=\mathrm{IN} \wedge \exists w \in N(v):w.state=\mathrm{IN}\\
		\end{aligned}$
		\State $\begin{aligned}[t]&hesitate(v) \equiv P(v) \not = \emptyset \wedge \exists w \in N(v):\big[ w \in P(v)\\& \wedge\forall z\in N(v)\setminus\{w\}:v\not\in P(z) \\&\wedge \forall z \in N(w) \setminus \{v\}:(z.state=\mathrm{OUT} \wedge w\in P(z)) \\&\wedge
		(v\not \in P(w)\vee w.id<v.id \vee \exists z\in N(w) \setminus \{v\}:w\in P(z))\big]
		\end{aligned}$
		\EndPredicates
		\Rules
		\State $\begin{aligned}[t]
		{R_1}:\ \ \ & pending(v)\wedge \neg hesitate(v) \xrightarrow{\text{\ \ $p$\ \ }} v.state:=\mathrm{IN}\\
		{R_2}:\ \ \  & conflict(v) \longrightarrow v.state:=\mathrm{OUT}\\
		{R_3}:\ \ \  &v.state=\mathrm{OUT} \wedge \exists w \in N(v)\cap P(v): (\neg pending(w) \\&\ \ \ \wedge w.state = \mathrm{OUT})\longrightarrow v.parents:= P(v)\setminus \{w\}\\
		{R_4}:\ \ \  &v.state=\mathrm{OUT} \wedge \exists w \in N(v)\setminus P(v): w.state = \mathrm{IN} \\&\ \ \ \longrightarrow v.parents:= P(v)\cup \{w\}\\
		{R_5}:\ \ \  &v.state=\mathrm{IN} \wedge \neg conflict(v)\wedge P(v) \not = \emptyset \\&\ \ \ \longrightarrow v.parent:= \emptyset\\
		{R_6}: \ \ \ & P(v)\setminus N(v)\not = \emptyset\longrightarrow v.parents:= P(v)\setminus N(v)\\
		{R_7}: \ \ \ & v.state=IN \wedge \neg conflict(v) 
		\xrightarrow{\text{\ \ $q$\ \ }} v.state:=OUT
		\end{aligned}$
		\EndRules
	\end{algorithmic}
\captionof{algocf}{deflection tolerant MIS ($dt$MIS)}
\label{fig:smis}
\end{algorithm}

Algorithm $dt$MIS is probabilistically self-stabilizing under a distributed fair scheduler. For the gain function \eqref{2_eq}, $dt$MIS is perturbation-proof and, according to $R_1$ and $R_7$, it tolerates violations and unauthorized actions, respectively, during convergence. For the most part, the proofs associated with $dt$MIS proceed along the same lines of $pf$MIS and $vt$MIS, and are thus omitted.

\red{The stochastic Bayesian game} for $dt$MIS is similar to the one for $vt$MIS except in the case of actions available to agent $v\in \mathcal{L}_{\mathfrak{a}}(k)$, which are
$$\scriptsize{A_v{(\lambda,\theta)}=\begin{cases}
	\{\textrm{Switch}, \textrm{Preserve}\}, &\textrm{if $pending(v)$ given $\lambda$ and $\theta$},\\ 
	\{\textrm{Switch}\}, &\textrm{if $conflict(v)$ given $\lambda$ and $\theta$}, \\ 
	\{\textrm{Switch, Preserve}\},  &\textrm{if $v.state=\mathrm{IN} \wedge\neg conflict(v)$ given $\lambda$ and $\theta$}, \\
	\{\textrm{Preserve}\}, &\textrm{otherwise}.
	\end{cases}}$$  

\subsection{Alternative Solutions (Special Cases)}

In this section, we present two alternative self-stabilizing algorithms that both work under an unfair distributed scheduler, for clustering in the face of selfish agents.

\subsubsection{Violation-Proof Solution (Special Case \uppercase\expandafter{\romannumeral 1})} 

 In some cases, when an agent refuses to execute a rule, the system reaches a situation where only that agent is enabled in the system. We call such a situation a dead-end.
  In a dead-end, assuming that the enabled agent profits more if it makes a move, it has no choice but to not violate.

A self-stabilizing algorithm is \emph{violation-proof} if it makes any violation-prone situation lead to a dead-end. Note that although existence of dead-ends does not contradict \Cref{t7}, a violation-proof algorithm does not need to account for behavior strategies. We do not model a dead-end situation by a game because there is no other enabled agent in the neighborhood.

\begin{theorem}\label{t_vp} Suppose a self-stabilizing DIS $Q$ where
	neither illegitimate configurations are Nash equilibria nor agents are able to selfishly execute actions.
	Assume that agents' self-interests lead to refusal of some rules.
	If any situation where 
	such rules are
	enabled is a dead-end in which the enabled agent profits more by making a move, 
the system is violation-proof.
\end{theorem}

\begin{proof}
	Assume an agent that violates a rule to prevent the system from converging to some configurations that are less profitable. Because such a rule is the only enabled rule in the neighborhood and the agent profits if it makes a move, it executes that rule. 
\end{proof}

In the event of a dead-end, one question is whether the enabled agent experiences the dead-end permanently if it continues to the violate? The agent cannot answer this question with certainty because its knowledge is limited to local information.
In response to this uncertainty, the agent may execute the enabled rule with a probability greater that 
the probability that the dead-end situation continues for any sequence of computations unless the enabled agent executes its enabled rule.  

We propose a violation-proof solution for MIS, $vp$MIS (Algorithm \ref{fig:vpmis}).

\begin{algorithm}
	\begin{algorithmic}
		\Constants
		\State $id$: Integer
		\EndConstants
		\Variables
		\State $state$: Binary
		\State $p_c$: Real \tcp{$p_c\in(0,1]$}
		\EndVariables
		\Predicates
		\State $\begin{aligned}[t]
		& d(v) \equiv |N(v)|\\
		& cmp(v,w) \equiv d(v)>d(w)\vee \big(d(v)=d(w)\wedge v.id<w.id\big)\\
		& pending(v) \equiv v.state=\mathrm{OUT} \wedge \forall w \in N(v):w.state=\mathrm{OUT}\\
		& conflict^*(v) \equiv v.state=\mathrm{IN} \wedge \exists w \in N(v):\big(w.state=\mathrm{IN} \ \wedge \\ & \ \ \ \ \ \ \neg cmp(v,w)\big)
		\end{aligned}$
		\EndPredicates
		\Rules
		\State $\begin{aligned}[t]
		{R_1}: \ \ \ & pending(v) \wedge \forall w \in N(v):\big(\neg pending(w) \vee cmp(v,w) \big) \\& \ \ \  \xrightarrow{\text{\ \ $p_c$\ \ }} v.state:=\mathrm{IN}\\
		{R_2}: \ \ \ & conflict^*(v) \longrightarrow v.state:=\mathrm{OUT}\\
		\end{aligned}$
		\EndRules
	\end{algorithmic}
	\captionof{algocf}{violation-proof MIS ($vp$MIS)}
	\label{fig:vpmis}
\end{algorithm}

Let $d(v)$ refers to the number of neighbors of agent $v$ and let $cmp(v,w)$ returns a single True or False value based on the outcome of comparison between the number of neighbors of $v$ and $w$ and in case of equality, their unique identifiers. Then, predicate $conflict^*(v)$ decides whether or not the state of agent $v$ is IN and there exists at least one IN neighbor of agent $v$ that either its degree is more than the degree of $v$ or both agents have the same degrees but it has a smaller identifier. 

The guard of the first rule ($R_1$) is equivalent to the event of a dead-end (i.e., if an agent is enabled by $R_1$, none of its neighbors are enabled) with an element of uncertainty; therefore, $R_1$ is probabilistic. It updates the state of agent $v$ to IN with a probability $p_c$ if there exists no neighbor of $v$ like $w$ such that either the state of $w$ is IN or it has no IN neighbor and $cmp(v,w)$ is false. The element of uncertainty arises only if $v$ has at least one neighbor $w$ such that $cmp(w,v)$ is true. In this case, if $v$ is enabled by $R_1$, $w$ has at least one IN neighbor; however, this is not necessarily true in the next computations. Therefore, given that $v$ is enabled by $R_1$, the value of $p_c$ is greater than the probability that any agent $w \in N(v)$ that satisfies $cmp(w,v)$ has at least one IN neighbor in the next rounds, i.e., the probability that $v$ remains enabled by $R_1$ forever if the state of $v$ never changes. The second rule updates the state of $v$ to OUT If $conflict^*(v)$ evaluates to true.

\begin{lemma}\label{lec1a1}If an agent is enabled by $R_1$, none of its neighbors are enabled.\end{lemma}
\begin{proof}Suppose an agent $v$ is enabled by $R_1$ and let $w$ be a neighbor of $v$. Agent $w$ cannot be enabled by $R_1$ because it has a neighbor, $v$, such that $pending(v)$ is true while $cmp(w,v)$ is false. It is not enabled by $R_2$ either because the state of $w$ is OUT.\end{proof}

\begin{lemma}\label{lec1a2}If no agent is enabled in the system, the configuration is legitimate.\end{lemma}
\begin{proof}Suppose that no agent is enabled but the MIS is not constructed. So, there exists at least two neighboring agents $v$ and $w$ that both are cluster-heads or there is at least one agent $v$ that has no IN neighbor. In the first case, $R_2$ is enabled either for $v$ or $w$ that contradicts our assumption. In the second case, the state of neighbors of $v$ must be OUT and thus either $v$ or one of its neighbors will execute $R_1$ to enter into the IS, which is a contradiction.\end{proof}

\begin{theorem}\label{tc1a1} $vp$MIS is self-stabilizing under an unfair distributed scheduler.\end{theorem}
\begin{proof}Suppose an agent $v$ is enabled in the initial configuration. If it is enabled by $R_1$, according to \Cref{lec1a1}, none of its neighbors are enabled. Given that $p_c$ is non-zero and $v$ remains enabled, the action of $R_1$ eventually updates the state of $v$ to IN. After executing $R_1$, the neighbors of $v$ will not execute any rule because they has an IN neighbor. Furthermore, it is impossible that $v$ be henceforward enabled because it has no IN neighbor and its neighbors never change their states. Now, suppose $v$ is initially enabled by $R_2$ and thus has an IN neighbor $w$ such that $\neg cmp(v,w)$. After executing $R_2$, assuming that the state of $w$ becomes OUT and $v$ has no IN neighbor, either $v$ or one of its neighbors, but not both (\Cref{lec1a1}), will be enabled by $R_1$ and enters IS and then none of them execute any rule. Therefore, the system finally stabilizes (i.e., there is no enabled agent) into a configuration that according to \Cref{lec1a2}, is legitimate.\end{proof}

\begin{theorem}\label{tc1a2}
	Suppose that a system $Q$ executes $vp$MIS and the gain function is \eqref{2_eq}. Then, $Q$ is violation-proof.
\end{theorem}
\begin{proof}An agent may violate any rule that makes it cluster-head ($R_1$) in hopes that one of its neighbors enters IS. However, according to \Cref{lec1a1}, $R_1$ is a dead-end that can be either permanent or temporary. Let $p_c$ denotes the probability that it is permanent. Because violating $R_1$ can be non-profitable forever with a probability more than $p_c$, the agent will eventually execute $R_1$ after some hesitations.\end{proof}

\subsubsection{Deflection-Proof Solution (Special Case \uppercase\expandafter{\romannumeral 2})}



We call a self-stabilizing system deflection-proof if it ensures that no deflection is beneficial. Note that a deflection-proof solution prevents from violations and perturbations too. Next theorem suggests a method to devise a deflection-proof algorithm.  

\begin{theorem}\label{tdp}Assume that no illegitimate configuration of a self-stabilizing DIS $Q$ is a Nash equilibrium.
	The system $Q$ is deflection-proof assuming that there exists a legitimate configuration $l$ that irrespective of whether or not $v$ deflects, any sequence of computations reaches either $l$ or a dead-end that enforces execution of actions leading to $l$.
\end{theorem}

\begin{proof}
	Suppose $v$ violates an enabled rule or illegally executes an action. In both cases, the system reaches $l$ or a dead-end that makes $v$ execute the actions that lead to $l$. Therefore, $v$ does not benefit from unauthorized actions or violations.
\end{proof}





The rest of this section elaborates on a deflection-proof solution to clustering. We start with the introduction of an algorithm (\Cref{cl:al1}) that always builds a unique MIS irrespective of the initial configuration as follows. 

\begin{small}
	\begin{algorithm}[b]
		\SetKwInOut{Input}{Input}
		\SetKwInOut{Output}{Output}
		\Input{graph $G(V,E)$}
		\Output{set $I$}
		$I = \emptyset$\;
		\While{$V \not = \emptyset$}{\ForEach{$v \in V$}{$S = \{u.id \mid u \in \left(v \cup N\left(v\right)\right) \}$\;\If{$v = \arg \min(S)$}{$I = I \cup \{v\}$\;
				}}
			\ForEach{$v \in I$}{$V = V - (\{v\}\cup N(v))$\;}
		update $E$ according to $V$\;} 
		\caption{returns an unique MIS} 
		\label{cl:al1}
	\end{algorithm}
\end{small}

Let $V$ denotes a variable that is initially equal to the set of agents. At the first iteration, each agent (e.g., $v$) whose identifier is less than the identifiers of all of its neighbors becomes cluster-head and thus its neighbors never enter the IS. Therefore, $\{v\}\cup N(v)$ is subtracted from set $V$. These operations repeat iteratively until $V$ is equal to $\emptyset$.

Next, we present \Cref{fig:dpmis} as a self-stabilizing solution based on \Cref{cl:al1}. This algorithm works with an unfair distributed scheduler. Its sole predicate $liberated(v)$ evaluates true if agent $v$ has a neighboring cluster-head with a smaller identifier. \Cref{fig:dpmis} is simple enough that we would not bother to give a formal proof that it is self-stabilizing.
\begin{algorithm}
	\begin{algorithmic}
		\Constants
		\State $id$: Integer
		\EndConstants
		\Variables
		\State $state$: Binary
		\EndVariables
		\Predicates
		\State $\begin{aligned}[t]
		& liberated(v) \equiv \exists w \in N(v):(w.state=\mathrm{IN} \wedge w.id < v.id)
		\end{aligned}$
		\EndPredicates
		\Rules
		\State $\begin{aligned}[t]
		{R_1}: \ \ \ & v.state=\mathrm{OUT} \wedge \forall w \in N(v):\big(v.id<w.id \ \vee \\& \ \ \ liberated(w)\big)  \longrightarrow v.state:=\mathrm{IN}\\
		{R_2}: \ \ \ & v.state=\mathrm{IN} \wedge \exists w \in N(v):\neg liberated(w) \\& \ \ \  \longrightarrow v.state:=\mathrm{OUT}
		\end{aligned}$
		\EndRules
	\end{algorithmic}
	\captionof{algocf}{deflection-proof MIS ($dp$MIS)}
	\label{fig:dpmis}
\end{algorithm}

\begin{lemma}\label{lec2a1}A system that executes $dp$MIS has only one legitimate configuration.\end{lemma}
\begin{proof}
Suppose $Q$ runs $dp$MIS and has more than one legitimate configuration (e.g., $l_1$ and $l_2$). There is hence at least one agent $v$ that is cluster-head in $l_1$, but it is not cluster-head in $l_2$. We can conclude that $v$ has a neighbor (e.g. $w$) such that $w.id<v.id$. Because the state of $w$ is IN in $l_2$ and OUT in $l_1$, $w$ must have a neighbor with smaller identifier and thus we will reach the same conclusions for that neighbor of $w$. This arrangement of identifiers in descending order will continue until infinity, which contradicts our underlying assumption that the number of agents are finite.
\end{proof}

\begin{theorem}\label{tc2a2}
	Suppose that system $Q$ executes $dp$MIS and the gain function is \eqref{2_eq}. Then, $Q$ is deflection-proof.
\end{theorem}
\begin{proof}
According to \Cref{lec2a1}, $Q$ has a legitimate configuration (an MIS) $l$ that irrespective of deflections, any sequence of computations reaches to either $l$ or a permanent dead-end that enforces execution of actions that leads to $l$. Therefore, according to \Cref{tdp}, $Q$ is deflection-proof.
\end{proof}

Note that all MIS configurations of a system running $dp$MIS are illegitimate except one configuration which is also the only configuration that is a Nash equilibrium.
\section{Complexity Analysis of the Algorithms}
\label{section:analysisComplexity}

In \Cref{table:complexity}, we show the average-case time complexities of the proposed self-stabilizing algorithms under schedulers that best suits them. Parameters $n$ and $D$ are the number of nodes and network diameter, respectively. \red{In our analysis, we assume the degree of nodes is upper bounded.} Note that the complexities of probabilistic algorithms $vt$MIS, $dt$MIS, and $vp$MIS depend heavily on the probabilities that agents apply to enter or exit the IS. These probabilities may vary over time with the execution of the algorithms.
\red{Thus, we define random variables that indicate whether an agent executes a probabilistic rule, and analyze time complexities using their expectations.}


\begin{table*}
	\centering
	\begin{tabular}{c|c|c|c|c|c|}
			\cline{2-6}
			&$pf$MIS& $vt$MIS& $dt$MIS& $vp$MIS& $dp$MIS\\ 
			\hline
			\multicolumn{1}{ |c| }{scheduler}& \textit{central} & \textit{synchronous} & \textit{synchronous} & \textit{distributed} & \textit{distributed} \\
			\multicolumn{1}{ |c| }{convergence time 
			} & ${O}(n)$ & ${O}(\log n)$ & ${O}(\log n)$ & ${O}(D)$ & ${O}(\log n)$ \\
					\multicolumn{1}{ |c| }{state transitions
		}& ${O}(n)$ & ${O}(n)$ & ${O}(n)$ & ${O}(n)$ & ${O}(n)$ \\
			\hline 
	\end{tabular}
	\caption{Expected complexities of the proposed distributed self-stabilizing algorithms for the MIS problem.}
	\label{table:complexity}
\end{table*}

\subsection{$pf$MIS}

We analyze the time complexity of $pf$MIS under a central scheduler. A central scheduler evaluates only one enabled agent in each round. Note that this sequential selection of agents prevents new conflicts from happening. Now, assume an enabled agent $v$ is selected by the scheduler in round $t$. If $v$ has a conflict, it exits the IS and never has a conflict again. It can be easily shown that there is no conflict in the system after $n$ rounds; however, note that leaving the IS may cause a pending situation in another agent. \red{Therefore, there may be pending agents after the first $n$ rounds. We now show that during the second $n$ rounds, all pending situations are resolved and no new conflicts or pending situations occur.}
Then, if $v$ is enabled by the first rule, it enters the IS and its state remains unchanged. Bear in mind that if predicate $hesitate(v)$ is enabled in a pending agent $u$, although $u$ is not enabled, one of its neighbors is enabled by the first rule and thus enters the IS when it is its turn. Therefore, we can conclude that after $2n$ rounds, there is no enabled agent in the system. Hence, both round and move \red{(i.e., state transition)} complexities of $pf$MIS are $O(n)$ under a central scheduler. 

\subsection{$vt$MIS}

We analyze the time complexity of $vt$MIS under a synchronous scheduler. Under $vt$MIS, an agent executes the first rule
\red{with a probability that corresponds}
to strategy Preserve. \red{Without loss of generality, we assume that all agents enter the IS with 
a probability $\tilde{p}$.}
Therefore, in every round, any pending agent enters the IS with probability $\tilde{p}$ (Rule 1) and any agent that has a conflict leaves the IS (Rule 2) with probability one. 

\begin{lemma}
	\label{lemmaVTmis1}
	Assume an agent $v\in V$. Then, in every two consecutive rounds $t_i$ and $t_{i+1}$, one of these two cases is true:
	\begin{itemize}
		\item $v$ is stabilized, i.e., $v$ is a cluster-head that has no conflict, or $v$ has a neighboring cluster-head that has no conflict.
		\item $\exists w \in (v \cup N(v)): pending(w).$ 
	\end{itemize}
\end{lemma}
\begin{proof}
	Suppose agent $v$ is not stabilized in round $t$. Agent $v$ therefore is either pending, has a conflict, or is OUT and not pending but all of its IN neighbors have conflicts. \red{If $v$ is pending, according to the second case, the lemma holds self-evidently.} If it has a conflict, $v$ and its IN neighbors leave the IS and are thus pending in round $t+1$. Note that none of the OUT neighbors of an agent that has a conflict enters the IS. Finally, if $v$ has a neighbor $w\in N(v)$ that has a conflict, $w$ and its IN neighbors leave the IS and thus $w$ is pending in round $t+1$.   
\end{proof}

Suppose an undirected graph $G(V,E)$. Recall that whenever an agent $v$ becomes a cluster-head without arising a conflict, $v$ and its neighbors will not be enabled anymore. It can be easily verified that these agents no longer affect the other agents such that we can remove $v \cup N(v)$ and every edge adjusted to them from graph $G$, resulting a new graph $G'(V',E')$.
Let $X_{v}$ be a r.v. that takes value 1 if agent $v$ is removed and 0 otherwise,
$$
X_{v}= \begin{cases}1, & \text { if } v \in V \text{ is removed}, \\ 0, & \text { otherwise. }\end{cases}
$$

We define r.v. $X=\sum_{v \in V} X_{v}$ as the number of agents that are removed. Next, we prove that $\mathbb{E}[X]>0$ if the maximum degree is upper bounded as $|V|$ goes to infinity.

\begin{lemma}
	\label{lemmaVT}
	$\mathbb{E}[X] \geq \sigma|V|$, where $\sigma > c > 0$, if $G$ is a bounded degree graph.
	
\end{lemma} 

\begin{proof}
According to \Cref{lemmaVTmis1}, in every two consecutive rounds, for each agent $v\in V$ that is not stabilized yet, there is at least one pending agent $w \in v \cup N(v)$. If only $w$ and not its neighbors enter the IS, they (including $v$) never will be enabled. Let $d_w$ denotes the degree of $w$. The probability that agent $w$ enters the IS without a conflict is lower bounded by $\tilde{p} (1-\tilde{p})^{d_w}$. This is an inequality because not all neighbors of $w$ are necessarily pending. Let $G$ be a bounded degree graph with maximum degree $d$ (as $|V|$ goes to infinity), and let $p_E$ denote the minimum probability that an agent successfully enters the IS and stabilizes, then we have
$$\tilde{p} (1-\tilde{p})^{d} \leq p_E$$
Hence, in each round, every agent $v \in V$ stabilizes with a probability greater than ${p_E/2}$, which yields $\mathbb{E}[X] \geq \sigma|V|$, where constant $\sigma={\tilde{p} (1-\tilde{p})^{d}/2}$. Note that we divide $p_E$ by two because \Cref{lemmaVTmis1} holds for two consecutive rounds. 


\end{proof}

We aim to prove that the algorithm terminates in $O(\log n)$ rounds with a high probability.

Let $n_t$ denote the number of remaining agents at the beginning of round $t$. We define r.v. $Z_{t}$ as 
$$
Z_{t}= \begin{cases}1, & \text {if round $t$ removes at least $\sigma n_t$ agents,} \\ 0, & \text {otherwise.}\end{cases}
$$

Next, we find a lower bound for $\mathbb{E}\left[Z_{t}\right]$. Since $X$ satisfies $\Pr(X \leq a)=1$ for constant $n_t$, by applying the reverse Markov's inequality and following \Cref{lemmaDP}, we have
$$
\begin{aligned}
\mathbb{E}\left[Z_{t}\right]=&\Pr(Y > \sigma n_t/2) \\
& \geq \frac{\mathbb{E}[Y]-\sigma n_t/2}{n_t-\sigma n_t/2}\\&= \frac{\sigma n_t-\sigma n_t/2}{n_t-\sigma n_t/2}\\&=\frac{\sigma}{2-\sigma}
\end{aligned}
$$

Let $\chi$ denotes the number of rounds that $Z_t=1$. Then, after $T=(2-\sigma)c \log (n)/\sigma$ rounds, $\mathbb{E}\left[\chi=\sum_{t=1}^T Z_{t}\right] \geq c \log (n)$. Because the rounds are independent, we can apply Chernoff bound. Hence,
$$
\begin{aligned}
\Pr\left(\chi\leq(1-\delta) c \log \left(n\right)\right) & \leq e^{-\frac{\delta^{2} c \log n}{2}} \\
& \leq \frac{1}{n^{c'}}
\end{aligned}
$$

Therefore, $vt$MIS terminates in $O(\log n)$ rounds with a high probability. Next, we prove that the average-case move complexity is $O(n)$.

As we discussed earlier, when an agent $v$ joins the IS without arising a conflict, we can remove $v \cup N(v)$ and every edge adjusted to them from graph $G$. Let $V_t$ denote to the set of remaining agents in round $t$. Then, according to \Cref{lemmaVT}, the expectation of the number of remaining agents in the next round is $(1-\sigma)|V_t|$. Furthermore, each agent can make at most one move in each round but if it is stabilized (removed), it never makes a move. Let r.v. $M_t$ denote the total number of moves that agents have made until the end of round $t$. We have
\begin{equation}
\label{moveVTproof}
\begin{aligned}
M_t&\le\sum_{i=1}^{t}|V_i| \implies\\
\mathbb{E}[M_t]&\le \sum_{i=1}^{t}\mathbb{E}[|V_i|]\\
&\le \sum_{i=1}^{t}(1-\sigma)|V_i|\\
&\le n/\sigma
\end{aligned}
\end{equation}
Hence, the average-case move complexity is $O(n)$.

\subsection{$dt$MIS}

We analyze the time complexity of $dt$MIS under a synchronous scheduler. Similar to $vt$MIS, we assume that each agent enters the IS with \red{a probability
$\tilde{p}$.}
Regarding the seventh rule ($R_7$), if there is no conflict, we assume that each cluster-head $v$ leaves the IS with a probability $$\tilde{q}=\begin{cases}
0, & \text{if $\forall w\in N(v): hesitate(w) \vee \neg pending(w)$,}\\
\epsilon>0, & \text{otherwise.}
\end{cases}$$
Therefore, unlike $\tilde{p}$, the value of $\tilde{q}$ is characterized depending on whether $v$ can profit from a 1-fault IN-to-OUT or not. Finally, in each round, any agent that has a conflict definitely leaves the IS.

\begin{lemma}
	\label{DTR7lemma}
	Under a synchronous scheduler, a cluster-head that has no conflict at most one time leaves the IS (executes $R_7$) and then it stabilizes in the next round.
\end{lemma}
\begin{proof}
Assume a pending agent $v$. Considering the rules of $dt$MIS, if $v$ enters the IS without arising a conflict in round $t$, its neighbors put $v$ in their parent set (Rule 4) in round $t+1$. Then, in the next rounds, according to \Cref{subsec:approach:perturbation}, $v$ will never perturb ($q=0$). Note that $v$ may leave the IS in round $t+1$ but only $v$ and not its neighbors will make a move in the next round and return the IS again. 
\end{proof}

According to \Cref{DTR7lemma}, this once execution of $R_7$ has no effect on the other agents and thus it can increase the total number of rounds at most one (last round). Similarly, the number of moves increases at most $n$. The rest of analysis of $dt$MIS is the same as the analysis of $vt$, which implies that move and round complexities of $dt$MIS are $O(\log n)$ and $O(n)$, respectively.

\subsection{$vp$MIS}
\label{subsub:ComplexityvpMIS}


Without loss of generality, we assume that the distributed scheduler evaluates all of the enabled agents in each round and that the probability dedicated to the first rule (i.e., $p_c$) is one. We later remove these assumptions.

Assume an agent $v$ which is pending or has a conflict. If $v$ has a conflict, its conflict is resolved after at most one round and then $v$ never has a conflict again. Similarly, if $v$ is pending and is enabled by the first rule, it enters the IS and then its state remains unchanged. Finally, if $v$ is pending but it is not enabled, it has at least one pending neighbor (e.g., $w\in N(v)$) with a smaller $id$. Note that $w$ may not be enabled for the same reason. This sequence of agents, which are pending but not enabled, is at most $D-1$. Therefore, the final state of $v$ is determined after at most $D-1$ rounds.

Now, we consider that the distributed scheduler arbitrarily selects each agent with a probability $p_S>0$. Assume an agent $v$ whose first rule is enabled and has at least one neighbor with a smaller $id$ than it. Since probability $p_c$ corresponds to the uncertainty of $v$ about the permanency of the dead-end, we set $p_c=p_S$ so that in average $v$ waits $1/p_S$ rounds. Keep in mind that since the first rule is enabled, no neighbor of $v$ that has a smaller $id$ is pending, and the dead-end is temporary only if at least one of these neighbors (e.g., $w\in N(v)$) becomes pending in the next rounds. This happens only if the second rule is enabled in any neighbor of $w$ that is in the IS ($e.g., u\in N(w)$) due to a conflict with another agent in the IS. It takes in average $1/p_S$ rounds that the scheduler evaluates $u$ and thus $v$ randomizes its first rule with probability $p_S$.  


As we discussed earlier, no agent is enabled after at most $D$ rounds under the assumption that $p_c=1$ and $p_S=1$. Now, considering that the distributed scheduler evaluates each agent with probability $p_S$ in range $(0,1]$ and that each enabled agent enters the IS with probability $p_c=p_S$, the expectation of the number of rounds before convergence increases to $D/p_S^2$. Nevertheless, given that constant $p_S$ is independent of the number of agents, the round complexity of $vp$MIS is $O(D)$.

Note that if an agent $v$ enters the IS, it never executes an action again, and if it exists the IS, it never has a conflict again and thus the only action that it can execute in the next rounds is to enter the IS. Therefore, the number of moves that each agent may execute during the convergence is at most two. Hence, the move complexity of $vp$MIS is $O(n)$

\subsection{$dp$MIS}

Here, we first assume that the distributed scheduler selects all the agents in each round. Then, we will relax this assumption such that the scheduler selects each agent with a given probability. Suppose an undirected graph $G(V,E)$. In each round of algorithm $dp$MIS, any agent $v \in V$ whose identifier is less than the identifiers of all of its neighbors enters the IS (Rule 1) and its neighbors leave the IS (Rule 2). Agent $v$ and its neighbors then never will be enabled in the next rounds because none of their rules can be enabled anymore. Besides, because for each $w\in N(v)$, there is a neighbor with less identifier (i.e., $v$), it no longer matters whether or not the identifier of $w$ is less than its neighbors. Therefore, we can update sets $V$ and $E$ for the next round in the same way as lines 9-12 in \Cref{cl:al1}.

Note that an agent $v$ enters the IS when it has the smallest $id$ in $v \cup N(v)$. Therefore, it joints with probability $\frac{1}{d_{v}+1}$, where $d_v$ denotes the degree of $v$. 


Let $Y_{v u}$ be a r.v. that takes value 1 if edge $vu$ is removed and 0 otherwise,
$$
Y_{v u}= \begin{cases}1, & \text { if } v u \in E \text{ is removed}, \\ 0, & \text { otherwise. }\end{cases}
$$

We define r.v. $Y=\sum_{v u \in E} Y_{v u}/2$ as the number of edges that are removed.

\begin{lemma}
	\label{lemmaDP}
	$2 \mathbb{E}[Y] \geq|E|$.
\end{lemma} 

\begin{proof}
	We define r.v. $X_{v u}$ as
	$$
	X_{v u}= \begin{cases}1, & \text { if } \forall w \in N(v) \cup N(u) - \{v\}, \ v.id < w.id, \\ 0, & \text { otherwise, }\end{cases}
	$$
	where $v u \in E$.
	
	Note that if $X_{v u}=1$, then $v.state=\mathrm{IN}$ and $u.state=\mathrm{OUT}$. The probability that $v$ has the smallest $id$ among $d(v)+d(u)$ agents is $1 /\left(d_{v}+d_{u}\right)$. Therefore, given that $v$ and $u$ may share agents, we have
		\begin{equation}
	\label{eq:proofAnDp1}\mathbb{E}\left[X_{v u}\right]
	\geq 1 /\left(d_{v}+d_{u}\right).\end{equation}
	
	Now, assume an agent $u \in V$. Note that at most one $X_{v u}=1$, and $\forall v \in N(u)$, if $X_{v u}=1$, then $\forall w \in N(u): Y_{w u}=1$. Thus,
	$$
	\sum_{v \in N(u)} d_{u} X_{v u} \leq \sum_{v \in N(u)} Y_{v u}
	$$
	which yields
	\begin{equation}
	\label{eq:proofAnDp2}
		\sum_{v u \in E} d_{u} X_{v u} \leq \sum_{vu \in E} Y_{v u}=2 Y.
	\end{equation}
	
	Following \eqref{eq:proofAnDp1} and \eqref{eq:proofAnDp2}, we have 
	$$
	\begin{aligned}
	2 \mathbb{E}[Y] & \geq \mathbb{E}\left[\sum_{v u \in E} d_{u} X_{v u}+d_{v} X_{u v}\right] \\
	&=\sum_{v u \in E} d_{u} \mathbb{E}\left[X_{v u}\right]+d_{v} \mathbb{E}\left[X_{u v}\right]\\
	&\geq \sum_{v u \in E} \frac{d_{u}}{d_{v}+d_{u}}+\frac{d_{v}}{d_{v}+d_{u}} \\
	&=\sum_{v u \in E} 1 \\
	&=|E|.
	\end{aligned}
	$$
\end{proof}

Let $m_t$ denote the number of remaining edges at the beginning of round $t$. Then, we define random variable $Z_{t}$ as 
$$
Z_{t}= \begin{cases}1, & \text {if round $t$ removes at least $m_t/3$ edges,} \\ 0, & \text { otherwise.}\end{cases}
$$

Next, we find a lower bound for $\mathbb{E}\left[Z_{t}\right]$. Since $Y$ satisfies $\Pr(Y \leq a)=1$ for constant $m$, by applying the reverse Markov's inequality and following \Cref{lemmaDP}, we have
$$
\begin{aligned}
\mathbb{E}\left[Z_{t}\right]=&\Pr(Y > \frac{m}{3}) \\
& \geq \frac{\mathbb{E}[Y]-\frac{m}{3}}{m-\frac{m}{3}}\\&= \frac{\frac{m}{2}-\frac{m}{3}}{m-\frac{m}{3}}\\&=\frac{1}{4}
\end{aligned}
$$

Given that the distributed scheduler arbitrarily selects each agent with a probability $p_S>0$, $\mathbb{E}\left[Z_{t}\right] \ge \frac{p_S}{4}$ instead.

Let $\chi$ denotes the number of rounds that $Z_t=1$ is satisfied. Then, after $T=4 c \log (n)/p_S$ rounds, $\mathbb{E}\left[\chi=\sum_{t=1}^T Z_{t}\right] \geq c \log (n)$. Because identifiers are randomly assigned to agents, the rounds are independent and thus we can apply Chernoff bound. Hence,
$$
\begin{aligned}
\Pr\left(\chi\leq(1-\delta) c \log \left(n\right)\right) & \leq e^{-\frac{\delta^{2} c \log n}{2}} \\
& \leq \frac{1}{n^{c'}}
\end{aligned}
$$

Therefore, $dp$MIS terminates in $O(\log n)$ rounds with a high probability. Note that the worst-case round complexity is $O(n)$, which is obtained if in every round, the number of agents that join the IS is upper bounded by a constant number (see \Cref{fig:analysis_1}). Observe that in each round, at least one agent joins the IS, which is the agent with the smallest identifier.

\begin{figure}[!h]
	\centering
	\includegraphics[width=0.4\textwidth]{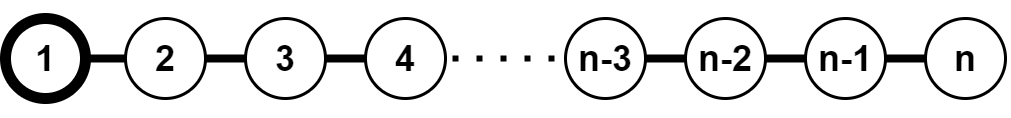}
	\captionof{figure}{An example of the worst-case round complexity where $dp$MIS after $n-1$ rounds terminates.}
	\label{fig:analysis_1}
\end{figure}

The average-case move complexity of $dp$MIS is $O(n)$. The proof is similar to \eqref{moveVTproof}. Note that the worst-case move complexity is $O(n^2)$ (see \Cref{fig:analysis_2}).

\begin{figure}[!h]
	\centering
	\includegraphics[width=0.4\textwidth]{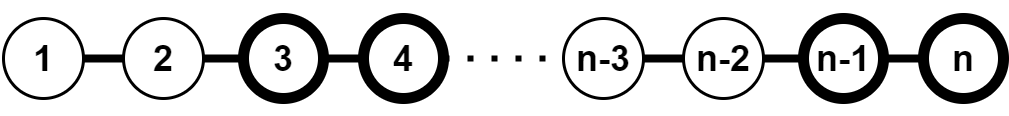}
	\captionof{figure}{An example of the worst-case move complexity where $dp$MIS after $n^2$ moves terminates.}
	\label{fig:analysis_2}
\end{figure}

\section{Evaluation Results}
\label{section:results}

In this section, we examine the fairness, time complexity, the ability to contain faults, defeat deviations during convergence and deal with selfish agents, and the outcomes of $b$MIS, the five algorithms proposed in \Cref{section:casestudy} ($pf$MIS, $vt$MIS, $dt$MIS, $vp$MIS, $dp$MIS), $t$MIS \cite{turau_linear_2007}, and $a$MIS \cite{arapoglu2019asynchronous} through simulation. 
We include algorithms $t$MIS and $a$MIS because they have linear time complexities. Similar to $vp$MIS and $dp$MIS, both $t$MIS and $a$MIS work under an unfair distributed scheduler.
In $t$MIS, a new state value WAIT is devised to prevent nodes from simultaneously entering the IS, so that among the WAIT nodes, a node with the smallest ID enters.
Unlike $t$MIS, $a$MIS uses two-hop information and stabilizes after at most $n-1$ moves where $n$ is the number of nodes.

\subsection{Experimental Setup}

We consider several scenarios involving different density networks. Simulations are run with MATLAB and each scenario is repeated at least 10000 times, i.e., 100 random initial configurations of 100 random scale-free connected undirected graphs generated using the preferential attachment graph model \cite{barabasi1999emergence}. We also employ the Erd\H{o}s-R\'enyi model \cite{erdhos_evolution_1960}, but do not report the results because they are similar to those obtained for the other graph model.

We run algorithms under a distributed randomized scheduler.  
We define \red{\textit{synchrony} as an indicator of how well the scheduler simultaneously selects agents, which is equivalent to}
the probability that the scheduler selects an agent to make a move during a round.

We use the gain function of \eqref{2_eq} with $\vartheta=10$ and $\zeta=1$. 
To simplify the model, we assume that each agent always knows its exact $2$-local state.
Under this assumption, the game models that we use for $vt$MIS and $dt$MIS are the stochastic Bayesian games given in Sections~$\ref{subsection:vtmis}$ and $\ref{subsection:dtmis}$, respectively. We solve these games using a
combination of the Bellman equation \cite{bellman1957dynamic} and the Bayesian Nash equilibrium. We assign 0.88 to the discounted factor parameter used in the calculations.

We set the value of $p$ in $vp$MIS to \red{synchrony} (see \Cref{subsub:ComplexityvpMIS}).


\subsection{Fairness Analysis}

We measure the fairness of algorithms using Jain's index \cite{jain1984quantitative}, which yields an index between 0 and 1. The fairness index is maximized when all agents receive the same profit. \Cref{table:ff} accounts for three different scenarios wherein we evaluate the fairness of algorithms.

\begin{table}
	\centering
	\scalebox{1}{\begin{tabular}{ |c||c|c|c|  }
			\hline
			scenario&\# of agents&avg. degree&\red{synchrony}\\
			\hline
			sparse &50&4&   0.7\\
			medium &50&6&   0.7\\
			dense&500& 24   &0.7\\
			\hline
	\end{tabular}}
	\caption{The specification of three scenarios}
	\label{table:ff}
\end{table}


Results are illustrated in \Cref{fig:fairness}. We observe that $vp$MIS, which selects cluster-heads primarily based on the number of neighbors, has the smallest fairness index, followed by  $dp$MIS, $a$MIS, and $t$MIS. 
The reason is that in these algorithms, some agents are always more likely to become cluster-heads (e.g. ones that have smaller identifiers) and thus these algorithms are more unfair to such agents.
We also notice that in denser networks, the fairness index is lower.
Given that an agent with a large degree is less likely to become a cluster-head than one with a small degree, as the network becomes denser, the difference between the degrees increases and thus the fairness index decreases.


\begin{figure}[!h]
	\centering
	\includegraphics[width=0.4\textwidth]{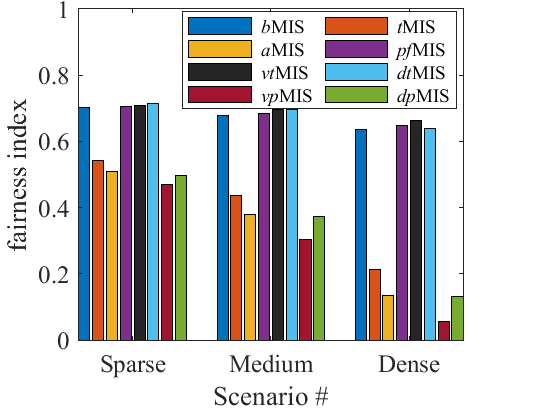}
	\captionof{figure}{Fairness of the algorithms in three scenarios.}
	\label{fig:fairness}
\end{figure}

\subsection{Time Complexity Analysis}

Next, we report results regarding time complexity in terms of moves and rounds.

Figures \ref{fig:rm}(a) and \ref{fig:rm}(b) show the average number of moves and rounds, respectively, for the algorithms as we vary the number of agents from 10 to 200.



The results demonstrate that the algorithms that work under unfair distributed scheduler have the lowest number of of moves (i.e., state transitions). These algorithms limit the number of moves by limiting the enabled agents to the ones with smaller identifiers. Among the other algorithms, $vt$MIS has the smallest number of moves. This algorithm causes an agent to enter IS with a probability. As a result, the simultaneous entry of neighbors into the IS is reduced. This decreases the number of conflicts, which leads to faster convergence and fewer moves. In contrast, $dt$MIS exhibits the worst performance because it suffers from a rule that allows agents to exit an MIS while there is no conflict. We also observe that large degrees have an extremely adverse impact on the convergence time of $dt$MIS (not shown in the figure) because as the degree of a cluster-head increases, the competition among it and its neighbors increases and thus the probability that it selfishly exits the IS increases. 

Figures \ref{fig:rm2}(a) and \ref{fig:rm2}(b) assess the impact of \red{synchrony} on the performance of algorithms. As \red{synchrony} increases, the number of agents that execute actions in a round increases; therefore, the system converges faster to the MIS. In particular, the convergence time of $vp$MIS, which assigns \red{synchrony} to the probability of its first rule, sharply drops as \red{synchrony} increases. On the other hand, the number of conflicts due to simultaneous neighbor actions increases and thus the number of moves increases, which especially affects $b$MIS and $pf$MIS when \red{synchrony} is close to one. Note that these two algorithms cannot work with a synchronous scheduler.

\begin{figure}[!h]
	\centering
	\begin{tabular}{c}
		
	\includegraphics[width=0.4\textwidth]{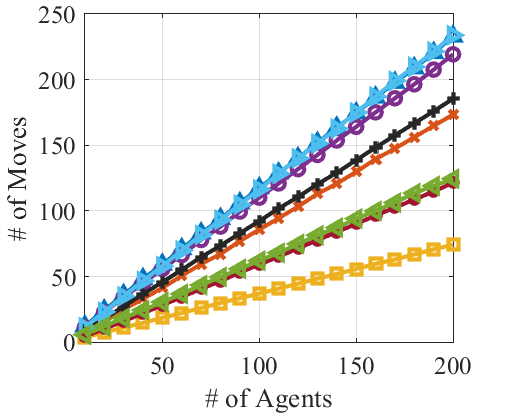} \\
	(a)\\
	\includegraphics[width=0.4\textwidth]{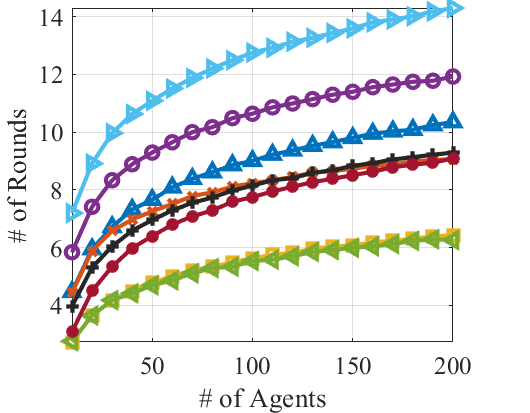} \\
	(b)\\
\includegraphics[width=0.45\textwidth]{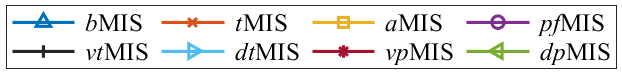}
\end{tabular}
	\caption{Avg. \# of $(a)$ moves and $(b)$ rounds versus \# of agents when \red{synchrony} is 0.7 and avg. degree is 8.}
\label{fig:rm}
\end{figure}


\begin{figure}[!h]
		\centering
	\begin{tabular}{c}
		\includegraphics[width=0.4\textwidth]{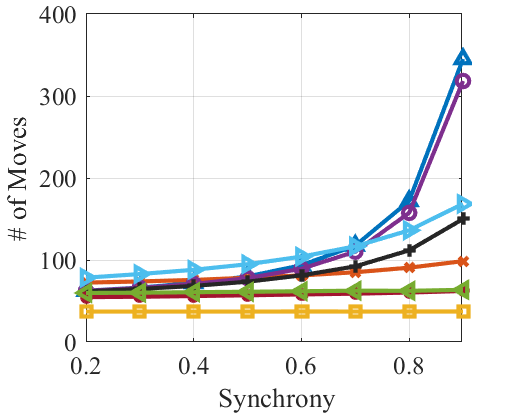}\\
		(a)\\
		\includegraphics[width=0.4\textwidth]{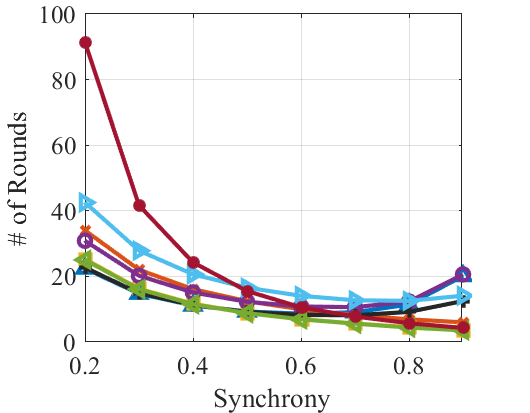}\\
		(b)\\
		\includegraphics[width=0.45\textwidth]{plots/new/legend.png}
	\end{tabular}
	\caption{Avg. \# of $(a)$ moves and $(b)$ rounds versus \red{synchrony} when \# of agents is 100 and avg. degree is 8.}
	\label{fig:rm2}
\end{figure}

\subsection{Fault-Containment Analysis}

Next, we assess the impact of a state change of a cluster-head from IN to OUT in a legitimate configuration (i.e., an IN-to-OUT 1-fault ) on the system.
Table \ref{table:1f} presents the number of moves and rounds needed to return to a legitimate configuration and the 1-fault success rate (i.e., the probability that the state of the faulty agent remains OUT after re-convergence) for the algorithms.

 We observe that $pf$MIS and $dt$MIS, which are meant to handle an IN-to-OUT 1-fault  in the faulty agent without propagating the fault to other agents, return to legitimate configurations in only one move. Also, the average numbers of moves that $a$MIS and $vp$MIS make to recover from an IN-to-OUT 1-fault  are almost one because the faulty agent usually enters a dead-end situation. Furthermore, despite the similarity of the rules in $b$MIS and $vt$MIS, $vt$MIS significantly reduces the error propagation in the network more than $b$MIS. Moreover, same as $pf$MIS and $dt$MIS, $dp$MIS yields a zero success rate because it has only one legitimate configuration. 

\begin{table*}
\centering
\begin{tabular}{ c|c|c|c|c|c|c|c|c|}
\cline{2-9}
&$b$MIS&   $t$MIS&  $a$MIS& $pf$MIS& $vt$MIS& $dt$MIS& $vp$MIS& $dp$MIS\\ 
\hline
\multicolumn{1}{ |c| }{  Avg. \# of moves }& 2.93 & 4.14 & 1.05 & 1.00 & 2.34 & 1.00 & 1.01 & 4.57 \\
\multicolumn{1}{ |c| }{Avg. \# of  rounds}& 2.02 & 3.15 & 1.28 & 1.25 & 1.87 & 1.25 & 1.53 & 2.28 \\
\multicolumn{1}{ |c| }{1-fault success rate}& 0.20 & 0.10 & 0.02 & 0.00 & 0.21 & 0.00 & 0.05 & 0.00\\
\hline 
\end{tabular}
\caption{The performance of system against an IN-to-OUT 1-fault  and the success rate of that perturbation when the number of agents is 40, average degree is 6, and \red{synchrony} is 0.8.}
\label{table:1f}
\end{table*}

\begin{table*}[]
	\centering
	\begin{tabular}{cc|c|c|c|c|c|c|c|c|c|}
		\cline{3-11}
		&  & $b$MIS & $t$MIS & $a$MIS & $pf$MIS & $vt$MIS & $dt$MIS & $vp$MIS & $vp$MIS$^*$ & $dp$MIS \\ \hline
		\multicolumn{1}{|c|}{Dense} & \# of attempts & 8951 & 13441 & 13458 & 8378 & 7785 & 8936 &7371& 7571 & 26341 \\[0.6ex]  \cline{2-2}
		\multicolumn{1}{|c|}{Network} &  succ. rate & 0.91 & 0.70 & 0.32 & 0.67 & 0.70 & 0.73 &0.11& 0.10 & 0.00 \\[0.3ex]  \hline\hline
		\multicolumn{1}{|c|}{\centering Sparse } & \# of attempts & 34350 & 37120 & 36220 & 30698 & 29905 & 28400 &26580& 26618 & 49853 \\[0.6ex] \cline{2-2} 
		\multicolumn{1}{|c|}{Network} &  succ. rate & 0.54 & 0.37 & 0.18 & 0.33 & 0.38 & 0.39 &0.04& 0.04 & 0.00 \\[0.3ex]  \hline
	\end{tabular}
\caption{\# of violations during convergence and the success rate of agents to change their fate by these violations.}
	\label{table:case1}
\end{table*}

\begin{table*}[]
	\centering
	\begin{tabular}{cc|c|c|c|c|c|c|c|c|c|}
		\cline{3-11}
		&  & $b$MIS & $t$MIS & $a$MIS & $pf$MIS & $vt$MIS & $dt$MIS & $vp$MIS & $dp$MIS & $R_7$ ($dt$MIS)\\ \hline
		\multicolumn{1}{|c|}{Dense} & \# of attempts & 1135 & 23660 & 38044 & 1135 & 1135 & 1045 & 13326 & 20607 & 51 \\[0.6ex]  \cline{2-2}
		\multicolumn{1}{|c|}{Network} &  succ. rate & 0.17 & 0.63 & 0.62 & 0.17 & 0.17 & 0.19 & 0.86 & 0.00 & 0.47 \\[0.3ex]  \hline\hline
		\multicolumn{1}{|c|}{\centering Sparse } & \# of attempts & 20013 & 50992 & 62899 & 20013 & 20013 & 19841 & 37550 & 42755 & 1195 \\[0.6ex] \cline{2-2} 
		\multicolumn{1}{|c|}{Network} &  succ. rate & 0.15 & 0.41 & 0.42 & 0.15 & 0.16 & 0.16 & 0.52 & 0.00 & 0.50 \\[0.3ex]  \hline
	\end{tabular}
\caption{\# of unauthorized actions (the seventh rule of $dt$MIS in the case of last column) during convergence and the success rate of agents to change their fate by these deviations.}
\label{table:case2}
\end{table*}

\begin{table*}
	\centering
	\begin{tabular}{ cc|c|c|c|c|c|c|c|c|c|  }
		\cline{3-10}
		&&\multicolumn{8}{ |c| }{Reliability}\\
		\cline{2-10}
		&\multicolumn{1}{ |c| }{Dev. type}& {$b$MIS}&    {$t$MIS}&  {$a$MIS}&  {$pf$MIS}&    {$vt$MIS} &  {$dt$MIS} &  {$vp$MIS} &  {$dp$MIS}\\
		\hline
		\multicolumn{1}{ |c  }{\multirow{5}{*}{\rotatebox[origin=c]{90}{Sparse}}}&\multicolumn{1}{ |c| }{No Deviation}& 1.00 & 1.00 & 1.00 & 1.00 & 1.00 & 1.00 & 1.00 & 1.00 \\
		\multicolumn{1}{ |c  }{}&\multicolumn{1}{ |c| }{Perturbation}& 0.00 & 0.00 & 0.00 & 1.00 & 0.00 & 1.00 & 0.00 & 1.00 \\
		\multicolumn{1}{ |c  }{}&\multicolumn{1}{ |c| }{Violation}& 0.00 & 0.05 & 0.18 & 0.00 & 1.00 & 1.00 & 1.00 & 1.00 \\
		\multicolumn{1}{ |c  }{}&\multicolumn{1}{ |c| }{Deflection}& 0.00 & 0.00 & 0.00 & 0.00 & 0.00 & 1.00 & 0.00 & 1.00 \\ 
		\hline 
		\hline 
		\multicolumn{1}{ |c  }{\multirow{5}{*}{\rotatebox[origin=c]{90}{Dense}}}&\multicolumn{1}{ |c| }{No Deviation}& 1.00 & 1.00 & 1.00 & 1.00 & 1.00 & 1.00 & 1.00 & 1.00 \\ 
		\multicolumn{1}{ |c  }{}&\multicolumn{1}{ |c| }{Perturbation}& 0.00 & 0.01 & 0.00 & 1.00 & 0.00 & 1.00 & 0.00 & 1.00 \\ 
		\multicolumn{1}{ |c  }{}&\multicolumn{1}{ |c| }{Violation}& 0.00 & 0.62 & 0.57 & 0.00 & 1.00 & 1.00 & 1.00 & 1.00 \\ 
		\multicolumn{1}{ |c  }{}&\multicolumn{1}{ |c| }{Deflection}& 0.00 & 0.00 & 0.00 & 0.00 & 0.00 & 1.00 & 0.00 & 1.00 \\ 
		\hline 
	\end{tabular}
\caption{Impact of deviation type on the reliability.}
\label{table:dev1}
\end{table*} 

\begin{table*}
	\centering
	\begin{tabular}{ cc|c|c|c|c|c|c|c|c|c|  }
		\cline{3-10}
		&&\multicolumn{8}{ |c| }{Average \# of Deviations}\\
		\cline{2-10}
		&\multicolumn{1}{ |c| }{Dev. type}& {$b$MIS}&    {$t$MIS}&  {$a$MIS}&  {$pf$MIS}&    {$vt$MIS} &  {$dt$MIS} &  {$vp$MIS} &  {$dp$MIS}\\
		\hline
		\multicolumn{1}{ |c  }{\multirow{5}{*}{\rotatebox[origin=c]{90}{Sparse}}}&\multicolumn{1}{ |c| }{No Deviation}& 0.00 & 0.00 & 0.00 & 0.00 & 0.00 & 0.00 & 0.00 & 0.00 \\
		\multicolumn{1}{ |c  }{}&\multicolumn{1}{ |c| }{Perturbation}& 462.24 & 269.37 & 360.69 & 0.00 & 449.97 & 0.00 & 495.56 & 0.00 \\
		\multicolumn{1}{ |c  }{}&\multicolumn{1}{ |c| }{Violation}& 3172.43 & 420.07 & 67.38 & 2959.24 & 0.00 & 0.00 & 0.00 & 0.00 \\
		\multicolumn{1}{ |c  }{}&\multicolumn{1}{ |c| }{Deflection}& 6688.98 & 1948.03 & 774.50 & 2968.43 & 458.40 & 0.00 & 516.94 & 0.00 \\
		\hline 
		\hline 
		\multicolumn{1}{ |c  }{\multirow{5}{*}{\rotatebox[origin=c]{90}{Dense}}}&\multicolumn{1}{ |c| }{No Deviation}& 0.00 & 0.00 & 0.00 & 0.00 & 0.00 & 0.00 & 0.00 & 0.00 \\
		\multicolumn{1}{ |c  }{}&\multicolumn{1}{ |c| }{Perturbation}& 444.21 & 151.92 & 266.86 & 0.00 & 341.26 & 0.00 & 116.12 & 0.00 \\
		\multicolumn{1}{ |c  }{}&\multicolumn{1}{ |c| }{Violation}& 4538.28 & 377.19 & 25.81 & 4791.34 & 0.00 & 0.00 & 0.00 & 0.00 \\
		\multicolumn{1}{ |c  }{}&\multicolumn{1}{ |c| }{Deflection}& 9159.00 & 3419.79 & 216.66 & 4778.70 & 343.07 & 0.00 & 113.96 & 0.00 \\
		\hline 
	\end{tabular}
\caption{Impact of deviation type on the average number of deviations.}
\label{table:dev2}
\end{table*} 

\begin{table*}
	\centering
	\begin{tabular}{ cc|c|c|c|c|c|c|c|c|c|  }
		\cline{3-10}
		&&\multicolumn{8}{ |c| }{Availability}\\
		\cline{2-10}
		&\multicolumn{1}{ |c| }{Dev. type}& {$b$MIS}&    {$t$MIS}&  {$a$MIS}&  {$pf$MIS}&    {$vt$MIS} &  {$dt$MIS} &  {$vp$MIS} &  {$dp$MIS}\\
		\hline
		\multicolumn{1}{ |c  }{\multirow{5}{*}{\rotatebox[origin=c]{90}{Sparse}}}&\multicolumn{1}{ |c| }{No Deviation}& 0.91 & 0.89 & 0.89 & 0.92 & 0.91 & 0.91 & 0.92 & 0.84 \\
		\multicolumn{1}{ |c  }{}&\multicolumn{1}{ |c| }{Perturbation}& 0.84 & 0.76 & 0.76 & 0.92 & 0.82 & 0.91 & 0.52 & 0.84 \\
		\multicolumn{1}{ |c  }{}&\multicolumn{1}{ |c| }{Violation}& 0.62 & 0.91 & 0.93 & 0.62 & 0.91 & 0.91 & 0.92 & 0.84 \\
		\multicolumn{1}{ |c  }{}&\multicolumn{1}{ |c| }{Deflection}& 0.20 & 0.58 & 0.35 & 0.62 & 0.81 & 0.91 & 0.51 & 0.84 \\
		\hline 
		\hline 
		\multicolumn{1}{ |c  }{\multirow{5}{*}{\rotatebox[origin=c]{90}{Dense}}}&\multicolumn{1}{ |c| }{No Deviation}& 0.93 & 0.90 & 0.88 & 0.94 & 0.92 & 0.91 & 0.94 & 0.86 \\
		\multicolumn{1}{ |c  }{}&\multicolumn{1}{ |c| }{Perturbation}& 0.84 & 0.83 & 0.68 & 0.94 & 0.75 & 0.91 & 0.59 & 0.86 \\
		\multicolumn{1}{ |c  }{}&\multicolumn{1}{ |c| }{Violation}& 0.51 & 0.87 & 0.91 & 0.51 & 0.92 & 0.91 & 0.94 & 0.86 \\
		\multicolumn{1}{ |c  }{}&\multicolumn{1}{ |c| }{Deflection}& 0.02 & 0.36 & 0.06 & 0.51 & 0.75 & 0.91 & 0.58 & 0.86 \\
		\hline 
	\end{tabular}
	\caption{Impact of deviation type on the availability.}
	\label{table:dev3}
\end{table*} 

\begin{table*}
	\centering
	\begin{tabular}{cc|c|c|c|c|c|c|c|c|}
		\cline{3-10}
		&  & $b$MIS & $t$MIS & $a$MIS & $pf$MIS & $vt$MIS & $dt$MIS & $vp$MIS & $dp$MIS \\ \hline
		\multicolumn{1}{|c|}{Same} & \# of configurations & 378 & 21 & 2 & 379 & 383 & 465 & 5 & 1 \\[0.6ex]  \cline{2-2}
		\multicolumn{1}{|c|}{Initialization} &  avg. \# of clusters & 17.42 & 19.30 & 20.00 & 17.09 & 16.53 & 16.71 & 10.24 & 15.00 \\[0.3ex]  \hline\hline
		\multicolumn{1}{|c|}{\centering Random } & \# of configurations & 1151 & 651 & 908 & 1253 & 1346 & 1342 & 468 & 1 \\[0.6ex] \cline{2-2} 
		\multicolumn{1}{|c|}{Initialization} &  avg. \# of clusters & 17.10 & 15.67 & 15.83 & 16.70 & 16.11 & 16.20 & 9.22 & 14.00 \\ \hline
	\end{tabular}
\caption{The outcome of the system in terms of the number of discovered legitimate configurations and average number of clusters, when \red{synchrony} is 0.7, the number of agents is 40, and average degree is 6.}
\label{table:sum}
\end{table*}

\subsection{Analysis of Deviations during Convergence}

Now, we study the outcome of the algorithms in the event of violations or deviations. To do so, during the convergence, we measure the effect of violating any rule that makes an agent cluster-head and the effect of executing an unauthorized action IN-to-OUT in terms of the number of attempts and success rate (i.e., the probability that the deviation succeeds to change the fate of the agent from a cluster-head to a cluster member after the system converges). We report the results regarding two different scenarios (\Cref{table:2}).

The results focusing on violations are shown in \Cref{table:case1}. The next-to-last column exceptionally shows the number of violations in $vp$MIS where the value of $p$ is set to $1/D$ instead of \red{synchrony} ($D$ is the diameter of the network and estimated as $\log(n)/\log(\log(n))$ \cite{bollobas2004diameter}). Although $vp$MIS causes an enabled OUT agent to confront a dead-end, the success rate of the violation is non-zero because the dead-end can be temporary. For example, in the sparse scenario, if an agent violates the first rule, it has a 4\% chance that the dead-end situation ends where one of its neighbors is cluster-head; therefore, it is rational that the agent commits a violation with a probability less than 0.03. We also observe that assigning $1/D$ to $p$ has little effect on reducing success rate while increasing round complexity to $O(D^2)$ (easily provable). Regarding $dp$MIS, we observe that an agent cannot change its fate with a violation or an unauthorized action. This is because the system has only one legitimate configuration. 

\Cref{table:case2} shows the results regarding unauthorized actions. The last column exceptionally shows the number of executions of $R_7$ in $dt$MIS and the success rate. We assume that a cluster-head executes an unauthorized action only if it has a neighbor that has no other neighboring cluster-head. We observe that the number of situations that an agent may execute such an action is several times bigger in the algorithms that discriminate between agents based on their identifiers.

\begin{table}
	\centering
	\begin{tabular}{ |c||c|c|c|  }
		\hline
		scenario&\# of agents& avg. degree &\red{synchrony}\\
		\hline
		dense network &20& 8   &0.7\\
		sparse network &20& 2&   0.9\\
		\hline
	\end{tabular}
	\caption{Specification of scenarios of dead-end tests}
	\label{table:2}
\end{table}

\subsection{Analysis of the Algorithms in the Presence of Selfishness}
In this section, we focus on how well the proposed algorithms react to selfish behaviors. We analyze the consequences of the three deviation types introduced in \Cref{section:approach} on the functionality of the algorithms with respect to sparse and dense connectivities (\Cref{table:dd}). 

\begin{table}
	\centering
	\scalebox{1}{\begin{tabular}{ |c||c|c|c|  }
		\hline
		scenario&\# of agents&avg. degree&\red{synchrony}\\
		\hline
		sparse &100&4&   0.8\\
		dense&100& 20   &0.8\\
		\hline
	\end{tabular}}
	\caption{The specification of sparse and dense connectivities.}
	\label{table:dd}
\end{table}


We define reliability as the probability that an algorithm succeeds in constructing an MIS before a time limit
of ten times its average convergence time 
when
agents do not exhibit selfishness.
Meanwhile, we count the number deviations that occur during convergence. We also measure the availability of valid clusters (i.e.,
 the average fraction of agents that either themselves or one of their neighbors are members of IS) during convergence.

Table \ref{table:dev1} reports reliabilities, numbers of deviations, and availabilities in the presence of each deviation type. All algorithms react differently to the different types of deviation. While $pf$MIS is reliable with respect to perturbations and $vt$MIS and $vp$MIS are reliable with respect to violations, $dt$MIS and $dp$MIS are reliable to the onset of all three types of deviations. We also observe that $t$MIS and $a$MIS can sometimes protect against violations. These two algorithms discriminate between agents, which lessens situations where an agent profits from violations. Furthermore, we observe that the number of deviations is inversely proportional to reliability and that it is significantly smaller in the case of perturbations. Moreover, as expected, 
if there is no deviation, all algorithms achieve the highest availability and in the case of deflections, they achieve the lowest availability.
Finally, we find that deflections have a very negative effect on availability in dense networks because when the degree of a cluster-head is high, it  is more likely to change its state by executing an unauthorized action.

\subsection{Analysis of Outcomes}

Table \ref{table:sum} reports the effect of each algorithm on the outcome of the system in terms of the number of unique discovered legitimate configurations and the average number of clusters when we repeat each test $10^4$ times. Concerning that the initial configuration is always the same, or it changes randomly in each repetition, the results are divided into two scenarios. 

Sure enough, $dp$MIS always terminates to the same legitimate configuration. Besides, the probability that two different runs of $a$MIS, $vp$MIS, or $t$MIS starting from the same configuration results in distinct legitimate configurations is very small. Furthermore, since $vp$MIS selects agents with larger degrees to be cluster-heads, it usually reaches legitimate configurations with the fewest number of clusters compared to the other algorithms, noting that a smaller number of clusters indicates a better result.

\section{Conclusion}
\label{section:conclusion}
	In this paper, we modeled rational interactions among selfish agents as a game and proposed three game-theoretic approaches for designing self-stabilizing algorithms in the face of those agents. We applied our solution methods to the problem of self-stabilizing clustering in a non-cooperative DIS. By setting the probabilities of self-stabilizing actions to behavior strategies, we proved that our first proposed algorithm tolerates any violations of the rules. Then, we proved that our second algorithm gives rise to Nash equilibria in legitimate configurations. Therefore, agents have no motivation to perturb clusters. Afterwards, we featured an algorithm that not only handles both violations and perturbations but also tolerates deflections during the convergence. We also introduced two algorithms intended to prevent deviations rather than tolerate them. The analysis of the results suggests that our solutions perform well concerning 
fairness and ability to deal with selfishness.

\bibliographystyle{IEEEtran}
\bibliography{paper.bib}

\end{document}